\documentclass[runningheads]{llncs}
\usepackage[T1]{fontenc}

\usepackage{href-ul}

\usepackage{appendix}

\usepackage{amsmath, amscd, amsfonts, amssymb, graphicx, color,amstext}
\usepackage{url}
\usepackage[shortlabels]{enumitem}

\usepackage{graphicx}
\usepackage{booktabs} 
\usepackage{bussproofs}
\usepackage{tikz}
\usepackage{listings}
\usepackage[linesnumbered,ruled,vlined]{algorithm2e}
\usepackage{adjustbox}
\usepackage{xspace}
\usepackage{tikz}
\usepackage{thmtools} 
\usepackage{thm-restate}
\usepackage{dsfont}

\EnableBpAbbreviations

\newcommand{\remove}[1]{}

\def\+#1{\mathcal{#1}}
\def\-#1{\mathbf{#1}}

\newcommand{\fik}{\textbf{FIK}\xspace}
\newcommand{\ik}{\textbf{IK}\xspace}
\newcommand{\ipl}{\textbf{IPL}\xspace}

\newcommand{\cfik}{$\mathbf{C}_{\fik}$\xspace}
\newcommand{\cik}{$\mathbf{C}_{\ik}$\xspace}
\newcommand{\ccik}{$\mathbf{CC}_{\ik}$\xspace}
\newcommand{\ccikann}{$\mathbf{CC}_{\mathbf{IK}}^{\text{n}}$\xspace}
\newcommand{\labik}{$\mathsf{labIK}_\leq$\xspace}
\newcommand{\lab}{\mathsf{lab}}

\newcommand\seq[2]{#1\Rightarrow #2}

\newcommand\iblock[1]{\langle#1\rangle}
\newcommand\mblock[1]{[#1]}
\newcommand{\rulebc}{($\text{inter}_{\sf bc}$)\xspace}
\newcommand{\rulefc}{($\text{inter}_{\sf fc}$)\xspace}
\newcommand\tr[1]{\mathsf{Tr}(#1)}
\newcommand\fm[1]{\mathsf{Fm}(#1)}
\newcommand\ant[1]{\mathsf{Ant}(#1)}
\newcommand\suc[1]{\mathsf{Suc}(#1)}

\newcommand{\str}{\subseteq^S}
\newcommand{\seqin}{\in^+}
\newcommand{\ibin}{\in^{\iblock{\cdot}}}
\newcommand{\mbin}{\in^{\mblock{\cdot}}}

\newcommand{\seqarrowann}[1]{\overset{#1}{\Rightarrow}}
\newcommand\aseq[3]{#1\overset{#2}{\Rightarrow}#3}
\newcommand\longaseq[3]{#1\overset{#2}{\Longrightarrow}#3}

\newcommand{\fc}{\mathsf{FC}}
\newcommand{\bc}{\mathsf{BC}}

\newcommand{\axfont}[1]{\mathsf{#1}}
\newcommand{\ax}[1]{\axfont{k}#1}
\newcommand{\modp}{\axfont{mp}}
\newcommand{\nec}{\axfont{nec}}

\newcommand{\hik}{\mathcal{H}_{\ik}}

\newcommand{\sq}[1]{#1}
\newcommand{\orseq}[1]{\langle #1 \rangle}
\newcommand{\rseq}[1]{[ #1 ]}

\newcommand{\orizero}[2]{{#1} \in^{\orseq{\cdot}}_0 {#2}}
\newcommand{\ori}[2]{#1\in^{\orseq{\cdot}}#2}
\newcommand{\rizero}[2]{{#1} \in^{\rseq{\cdot}}_0 {#2}}
\newcommand{\ri}[2]{#1\in^{\rseq{\cdot}}#2}
\newcommand{\allizero}[2]{#1\in^+_0 #2}
\newcommand{\alli}[2]{#1\in^+#2}
\newcommand{\stri}[3]{#1\subseteq^{#2}#3}

\newcommand{\stack}{\+A;~}

\newcommand{\proc}{$\mathsf{ProofSearch(A)}$\xspace}

\newcommand{\lr}[1]{#1_L}
\newcommand{\rr}[1]{#1_R}

\newcommand{\nikm}{\mathsf{NIKm}}
\def\lef#1{#1^\bullet}
\def\rig#1{#1^\circ}
\newcommand{\cseqp}[1]{G\{#1\}}
\newcommand{\fl}[1]{(#1)^\mathsf{t}}
\newcommand{\id}{\mathsf{id}}
\newcommand{\cseqnes}[1]{\Sigma\{#1\}}
\newcommand{\cseqpdown}[1]{\Sigma^{\downarrow}\{#1\}}
\newcommand{\tree}[1]{tr(#1)}
\newcommand{\cseqpstar}[1]{G^{*}\{#1\}}

\newcommand{\cseqnum}[2]{G_{#1}\{#2\}}
\newcommand{\seqstar}[1]{Q^{*}_{#1}}
\newcommand{\cstar}[1]{#1^{*}}
\newcommand{\wk}{\mathsf{wk}}

\newcommand{\dbikt}{\mathbf{DBiKt}}

\newcommand{\lbikt}{\mathbf{LBiKt}}

\newcommand{\dimpl}{{\mbox{$\;-\!\!\!<\;$}}}

\newcommand{\IMLs}{IMLs\xspace}

\begin{document}


\title{A Bi-nested Calculus for Intuitionistic $\mathbf{K}$: Proofs and Countermodels}

\author{
Han Gao\inst{1}, 
Marianna Girlando\inst{2}
and 
Nicola Olivetti\inst{1}
}

\institute{
LIS, Aix-Marseille University, CNRS;
\and
University of Amsterdam
}

\maketitle

\begin{abstract}
  The logic \ik is the intuitionistic variant of modal logic introduced by Fischer Servi, Plotkin and Stirling, and studied by Simpson. 
  This logic is considered a fundamental  intuitionstic modal system as it corresponds, modulo the standard translation, to a fragment of intuitionstic first-order logic. 
  In this paper we present a labelled-free bi-nested sequent calculus for \ik. This proof system comprises two kinds of nesting, corresponding to the two relations of bi-relational models for \ik: a pre-order relation, from intuitionistic models, and a binary relation, akin to the accessibility relation of Kripke models.    
  The calculus provides a decision procedure for \ik by means of a suitable proof-search strategy. 
  This is the first labelled-free calculus for \ik which 
  allows direct counter-model extraction: from a single failed derivation, it is possible to construct a finite countermodel for the formula at the root. 
  We further show the bi-nested calculus can simulate both the (standard) nested calculus and labelled sequent calculus, which are two best known calculi proposed in the literature for \ik. 
\end{abstract}


\section{Introduction}
\label{sec:introduction}

The world of intuitionistic modal logics (\IMLs in short) is richer than the classical one. Since the 50's, there have been several proposals  
for IML based on 
various semantics or proof-theoretical considerations, e.g. \cite{Wijesekera:1990}
\cite{Bellin:et:al:2001}, 
\cite{Fischer:Servi:1984} 
and 
\cite{Plotkin:Stirling:1986}.  


Among the various systems, a prominent role is played by \ik. This logic was introduced first by Fisher Servi as the `true' intuitionistic counterpart of classical modal logic \textbf{K}, then studied by Plotkin and Stirling and finally systematized by Simpson. 
In his thesis, 
Simpson proposed six criteria that can be used to classify intuitionistic variants of modal logic:  
\begin{enumerate}[$(i)$]
	\item  The logic should  be a conservative extension of intuitionistic logic; 
	\item  It should satisfy the disjunction property;
	\item  The two modalities $\Box$ and $\Diamond$ should not be 
    interdefinable; 
	\item  The two modalities should be `normal', in the sense of classical normal modal logic;  
	\item  Adding to the logic  any classical propositional tautology 
    not valid intuitionistically, 
 classical modal logic $\mathbf K$ should be obtained; 	 
		\item  The modal logic should be grounded on ``an  intuitionistic acceptable explanation of modalities''. Simpson  interprets this requirement as the fact  that the modal logic should correspond   to  a fragment of first-order intuitionistic logic by means of  the \emph{standard translation}. 		
\end{enumerate}
The logic $\ik$ (together with its  extensions with axioms from the standard modal cube) turns out to satisfy all of the criteria.
In contrast, the so-called constructive modal logic \textbf{CK} and alike (cf. \cite{Bellin:et:al:2001,Wijesekera:1990}) satisfy only the first three, whereas the logic $\fik$ recently introduced in \cite{csl-2024-fik} satisfies all criteria except the last one $(vi)$.

The semantics of $\ik$ 
can be specified in terms of bi-relational Kripke model equipped with a pre-order (denoted by $\leq$) taking care of intuitionistic implication and an \emph{accessibility} relation (denoted by $R$) for the modal operators. The two relations interact by means of two frame conditions named \emph{forward} and \emph{backward} confluence (originally named $\mathsf{F2}$ and $\mathsf{F1}$ respectively). Roughly speaking, the former condition is needed for preserving intuitionistic validity, whereas the latter is 
for first-order interpretability. 

Although the semantics of of $\ik$ is well-understood, the same cannot be said about its proof theory, as defining analytic calculi for $\ik$ has been a challenge for proof-theorists. 
While constructive modal logics enjoy simple cut-free Gentzen-style calculi (e.g. \cite{Wijesekera:1990}, \cite{Bierman:dePaiva:2000}), no cut-free sequent calculus system is known for \ik, 
and it seems unlikely that such a system exists.


Analytic proof systems for the logic (and other systems in the family) can be defined by enriching the structure of Gentzen-style calculi, for instance, single-conclusion nested sequent calculi were introduced in \cite{Galmiche:Salhi:2010,2013cut,marin2014label} while multi-conclusion nested sequents can be found in \cite{kuznets2019maehara}. 
In \cite{kuznets2019maehara}, a decision procedure for $\ik$ is provided, extracting a countermodel from failed proof search. However, since some of the rules in the calculus are not invertible, a single failed derivation might no be enough to construct a countermodel for the 
formula at the root. 
%
Alternatively, labelled calculi for $\ik$ and its extensions have been introduced in~\cite{simpson1994proof}, in the form of a single-conclusion labelled calculus, and in~\cite{Marin:et:al:2021}, where a fully labelled calculus was proposed, explicitly representing both semantic relations $R$ and $\leq$. 
Fully labelled calculi were used in \cite{morales:phd} to provide a countermodel construction for $\ik$; however, the procedure relies on a complex loopcheck, originally devised to establish decidability of intuitionistic modal logic $\mathbf{IS4}$, a logic stronger than $\mathbf{IK}$ (cf. \cite{Girlando:etal:2023}). A procedure tailored to \ik to extract finite countermodels from failed proof search in the  fully labelled calculus was proposed in \cite{wollic-2024-ik}. While this algorithm is much simpler than the one devised for $\mathbf{IS4}$, it still suffers from the usual difficulties encountered when devising termination strategies within labelled calculi. 
\footnote{Namely, to detect repetition of certain labels / worlds in the graph of labels generated by proof search, one first needs to `shrink' the graph, which otherwise would be growing indefinitely with always new labels.} 
Thus, it is worth looking into label-free proof systems to devise decision algorithms for \ik. 



In light of the above considerations, we are interested in proof systems that can provide both (a) a decision procedure for the logic 
, and (b) a  direct countermodel construction. By the latter we mean that  if \emph{one} derivation in the calculus of a formula $A$ fails, from \emph{that} derivation it is possible to extract a (finite) countermodel of $A$. In a constructive perspective,  we can say that if a proof of $A$ serves as a \emph{witness} of the validity of $A$, then a finite countermodel of $A$ acts as a \emph{witness} of its non-validity, and both 
play equally important roles. 
Neither (a) nor (b) are expected to be an easy task for \ik: the decision problem for this logic is not known to have an elementary upper bound, and countermodel construction is conjectured to have Ackerman function complexity, see e.g. \cite{Odintsov2007}. 

In this work we propose an innovative label-free calculus for \ik, called \cik. The calculus is `bi-nested', 
in the sense that it employs 
two kinds of nestings, intuitively encoding the two semantic relations $R$ and $\leq$. 
Typically, nested sequents encode only the 
accessibility 
relation $R$. 
In \cite{Fitting:2014}, Fitting proposed 
a nested system for intuitionistic logic internalising the $\leq$ relation instead, and in~\cite{anupam:2023} a similar proof system was introduced, capturing yet another variant of \ik\footnote{This logic, which we won't treat further, is characterized by $\ax 1, \ax2, \ax3$ and $\ax 5$ from Figure~\ref{fig:axioms} below.}. Besides, a preliminary  bi-nested calculus for \ik was 
suggested in \cite{Marin:Maroles:2020}. 
Our calculus \cik is an extension of the calculus for the logic \fik~mentioned above which was introduced in \cite{csl-2024-fik}.  
A special feature of bi-nested calculi is that the frame conditions of backward and forward confluence can be uniformly captured by suitable interaction rules operating on the two kinds of nesting. Specifically, \cik is obtained by adding to the proof system for \fik a rule corresponding to backward confluence. 

By adopting a suitable proof-search strategy, 
the calculus \cik provides a decision procedure for \ik. The proof system allows for a direct countermodel
construction, meaning that one 
failed 
branch of  
a 
derivation is sufficient to construct a countermodel for the root formula. Remarkably, the terminating proof search strategy defined in~\cite{csl-2024-fik} 
for 
\fik can be adapted to \cik. 
The loopcheck is simple as it only requires comparing sets of formulas (the worlds). 
The countermodel construction, however, is more complicated than the one for \fik. 
As usual, in our setting a countermodel is provided by a `finished' or `saturated' (non-axiomatic) sequent occurring as a leaf of any branch of a derivation and nested components are meant to correspond to worlds of the model. 
However, as a difference with the case of \fik, in order to ensure the \emph{backward confluence} property we cannot retain every component of a saturated sequent as a world in the model. Instead, we need an \emph{annotation} mechanism to 
designate 
which components will serve as worlds of the model. 
This makes the construction more complex, but still effective, and the size of a countermodel is not greater than the size of a saturated sequent.  



Moreover, it is instructive to compare our calculus \cik to other calculi for \ik proposed in the literature. We consider in detail the fully labelled calculus \labik and the (standard) nested sequent calculus $\nikm$. We aim to establish 
translations from both calculi to ours, in the sense that a derivation in each one of the two calculi can be effectively translated into a derivation in our calculus. Generally speaking, the interest of establishing translations or simulations between calculi is twofold.  
First, the fact that we can simulate a calculus $\mathsf{C}_1$ by a calculus $\mathsf{C}_2$ provides an `explanation' of the rules of $\mathsf{C}_1$ in terms of $\mathsf{C}_2$ independently from any semantic consideration. 
Second, relying on the completeness of $\mathsf{C}_1$ we get a purely syntactic proof of the completeness of $\mathsf{C}_2$, which again is independent from any semantic argument. A third reason for studying simulation would be to estimate the relative efficiency of the involved calculi by comparing the size of the original derivation in the source calculus with the size of its translation in the target calculus. The results presented here can be a starting point for investigating this topic for the related calculi. 

We first consider the fully labelled calculus \labik, we     show that every proof of an \ik formula in \labik  can be translated  into a proof  in our \cik. This result is not entirely obvious, since \labik is more expressive than \cik, making use of an enriched language and relational rules that have no counterpart in \cik. This translation is only possible for proofs of formulas, but not for arbitrary derivations in \labik. 


We then turn to a nested calculus $\nikm$, 
the multi-conclusion nested sequent calculus for $\ik$ from \cite{kuznets2019maehara}. 
We show that every derivation in $\nikm$ can be translated into a derivation in \cik. In this case the simulation is not straightforward for two reasons: first, although $\nikm$ and \cik are both nested sequent calculi, their syntax, and more importantly the notion of context, whence of of deep inference, is different. Furthermore, calculus $\nikm$ has no specific rules corresponding to the interaction rules in \cik (which encodes the semantic conditions of forward and backward confluence).  
The fact that we can simulate $\nikm$ by \cik `explains' the rules of $\nikm$ by decomposing them into more elementary steps of \cik.  

In addition, as mentioned before, the bi-nested calculus \cik is obtained by adding just one rule to the calculus of \fik presented in \cite{csl-2024-fik}. 
We believe this fact serves as a strong argument in favor of bi-nested calculi. The flexibility of this framework makes it well-suited for capturing various variants of intuitionistic modal logics in a unified way, offering both a decision procedure and countermodel extraction for these logics. 

The paper is structured as follows. In Section~\ref{sec:preliminaries} we introduce the logic \ik. Then, the calculus \cik is presented in Section~\ref{sec:calculus}. The terminating proof search strategy for \cik is presented in Section~\ref{sec:proof-serach}, while Section~\ref{sec:completeness} discusses semantic completeness. 
In Section \ref{sec:simulation}, we show how to simulate two known calculi for \ik into our \cik. 
We conclude with perspectives and future works in Section~\ref{sec:conclusions}. 

\section{The intutionistic modal logic \ik}
\label{sec:preliminaries}

We briefly introduce the
semantics and axiom system for the 
logic \ik. 
The set 
of formulas (denoted $A$, $B$, etc.) of our language is generated by the grammar
	%
	%
	$
	A::=~p~|~\bot~|~\top~|~(A\wedge A)~|~(A\vee A)~|~(A\supset A)~|~\Box A~|~\Diamond A
	$, 
	%
	%
	where $p$ ranges over  a  countable set of atomic propositions $\textsf{At}$. 
	Negation $\neg A$ is defined as $A\supset\bot$.
%
%
%
%
%
%


%
%
\begin{definition}
\label{bi-relational-model}
	A \textit{bi-relational model} is a quadruple $\+M=(W,\leq, R,V)$ where $W$ is a nonempty set of 
	elements, called \emph{worlds}, $\leq$ is  a reflexive and transitive relation (a \emph{pre-order}) over $W$, $R$ is a binary relation over $W$ and 
	the valuation function $V:~W\longrightarrow\wp(\textsf{At})$  
	satisfies 
	the 
	\emph{hereditary condition}:
	$
	\text{for all} x,y\in W, \text{ if } x\leq y \text{ then } V(x)\subseteq V(y).
	$  
	Moreover, $\+M$ satisfies the frame conditions of \emph{forward} and \emph{backward confluence}
    \footnote{In 
	\cite{simpson1994proof}, forward and backward confluence are known as $\mathsf{F1}$ and $\mathsf{F2}$ respectively. We take the terminology from 
	\cite{Balbiani:et:al:2021}.}
    (see Figure \ref{fig:fc:bc}):
	%
	%

%
%
\begin{itemize}
	%
	%
	\item[$(\fc)$]
	For all $x, x'\!, z\in W$, if $x\leq x'$\! and 
	$xRz$, 
	there is $z'\!\in W$ s.t.~$x'Rz'$\!
	and $z\leq z'$.
	\item[$(\bc)$]
	For all $x, z, z'\!\in W$,  
	if  
	$xRz$
	and $z\leq z'\!$, 
	there is $x'\!\in W$ s.t.~
	$x'Rz'$\!
	and $x\!\leq x'\!$.
	%
	%
\end{itemize}
In what follows, we shall sometimes write $x \geq y$ for $y \leq x$. 
\end{definition}

%
\begin{definition}
\label{definition:forcing:relation}
	Let 
	$\+M$ 
	be a bi-relational model. 
	The forcing conditions of a formula at a world $w\in W$ of $\+M$ are defined as follows:
	\begin{itemize}
        \item $\+M,w \not\Vdash \bot$ and  $\+M,w \Vdash \top$; 
		\item $\+M,w\Vdash p$~~iff~~$p\in V(w)$;


		\item $\+M,w\Vdash B\wedge C$~~iff~~$\+M,w\Vdash B$ and $\+M,w\Vdash C$;

		\item $\+M,w\Vdash B\vee C$~~iff~~$\+M,w\Vdash B$ or $\+M,w\Vdash C$;

		\item $\+M,w\Vdash B\supset C$~~iff~~for all $w'\in W$ with $w\leq w'$, if 
		$\+M,w'\Vdash B$, 
		then 
		$\+M,w'\Vdash C$;
		\item $\+M,w\Vdash \square B$~~iff~~for all $w',v'\in W$ with $w\leq w'$ and 
		$w'Rv'$,
		$v'\Vdash B$;
		\item $\+M,w\Vdash \Diamond B$~~iff~~there exists $v\in W$ with 
		$wRv$
		and $\+M,v\Vdash B$.
	\end{itemize}
	We 
	shall abbreviate $\+M,w\Vdash A$ as $w\Vdash A$ if the model is clear from the context.

	A formula $A$ in $\+L$ is \textit{valid}, denoted $\Vdash A$, if for any bi-relational model $\mathcal{M}$ and any world $w$ in it, it holds that $\mathcal{M},w\Vdash A$.
\end{definition}

The hereditary property can be extended to arbitrary formulas: for any $w, w'\in W$, for any formula $A$, if $w\Vdash A$ and $w\leq w'$ then $w'\Vdash A$. The proof of this property is established by induction on the complexity of $A$ and the case of $A=\Diamond B$ uses $(\fc)$. 


\begin{figure}\label{fig:fcbc}
    \begin{center}
        \begin{tikzpicture}[scale=0.8]
        \tikzstyle{node}=[circle,fill=black,inner sep=1.2pt]

        \node[label= left :{\small{$x$}}] (1) at (0,0) [node] {};
        \node[label= right :{\small{$z$}}] (2) at (2,0) [node]
        {};
        \node[label=  right :{\small{$z'$}}] (3) at (2,2) [node]
        {};
        
        \node[label= left :{\small{$x'$}}] (4) at (0, 2) [node]
        {};
        \node[][] at (1,1){\textsf{\textcolor{blue}{FC}}};

        \draw[->] (1) edge[below] node {\footnotesize{$R$}} (2);
        \draw[->, dotted] (2) edge[right] node{\footnotesize{$\leq$}} (3);
        \draw[->, dashed ] (1) edge[left] node{\footnotesize{$\leq$}}  (4);
        \draw[->,dotted,] (4) edge[above] node {\footnotesize{$R$}} (3);
        
        \end{tikzpicture}
        \quad 
        \begin{tikzpicture}[scale=0.8]
        \tikzstyle{node}=[circle,fill=black,inner sep=1.2pt]

        \node[label= left :{\small{$x$}}] (1) at (0,0) [node] {};
        \node[label= right :{\small{$z$}}] (2) at (2,0) [node]
        {};
        \node[label=  right :{\small{$z'$}}] (3) at (2,2) [node]
        {};
        
        \node[label= left :{\small{$x'$}}] (4) at (0, 2) [node]
        {};
        \node[][] at (1,1){\textsf{\textcolor{blue}{BC}}};

        \draw[->] (1) edge[below] node {\footnotesize{$R$}} (2);
        \draw[->, dashed] (2) edge[right] node{\footnotesize{$\leq$}} (3);
        \draw[->, dotted ] (1) edge[left] node{\footnotesize{$\leq$}}  (4);
        \draw[->,dotted,] (4) edge[above] node {\footnotesize{$R$}} (3);
        
        \end{tikzpicture}
        \vspace{-0.5cm}
        \end{center}
    \caption{Forward confluence (left) and backward confluence (right)
    }
    \label{fig:fc:bc}
\end{figure}
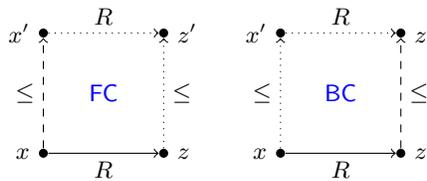

\begin{figure}
    \centering
    \small
    \begin{tabular}{c l @{\qquad}c}
        $\ax{1}$ & $\square(A\supset B)\supset(\square A\supset\square B)$ & 
        \\
        $\ax{2}$ &  $\square(A\supset B)\supset(\Diamond A\supset\Diamond B)$ & 
        \AxiomC{$A\supset B$}
	    \AxiomC{$A$}
	    \RightLabel{$\modp$}
	    \BinaryInfC{$B$}
	    \DisplayProof
        \\
        $\ax{3}$ & $\Diamond(A\vee B)\supset (\Diamond A\vee\Diamond B)$ & 
        \\
        $\ax{4}$ & $(\Diamond A\supset\square B)\supset\square (A \supset  B)$ & 
         \AxiomC{$A$}
	\RightLabel{$\nec$}
	\UnaryInfC{$\Box A$}
	\DisplayProof\\
        $\ax{5}$ & $\neg\Diamond\bot$ & \\
     \end{tabular}
     \normalsize
    \caption{Axioms and rules of $\ik$}
    \label{fig:axioms}
\end{figure}

The Hilbert-style axiom system for \ik, called $\hik$, is defined by adding to an axiomatization of intuitionistic propositional logic $\ipl$ the axioms and inference rules displayed in Figure~\ref{fig:axioms}. Soundness and completeness of the axiomatisation with respect to the semantics was proved  in~\cite{Fischer:Servi:1984}. 





\begin{theorem}
A formula $A$ is provable in $\hik$ if and only if $A$ is valid.
\end{theorem}
\section{A bi-nested calculus for \ik}\label{sec:calculus}

In this section we present a bi-nested sequent calculus, called \cik, for the logic \ik. The calculus makes use of two types of nesting, $\iblock{\cdot}$ and $\mblock{\cdot}$, 
respectively representing 
the pre-order $\leq$ and accessibility relation $R$ in bi-relational semantics. 
The two nestings give rise to two kinds of nested components, which we call \emph{implication block} and \emph{modal block} respectively. 
The calculus \cik contains two `interaction' rules capturing the frame conditions $(\fc)$ and $(\bc)$, which make implication and modal blocks interact with each other. 
And \cik can be regarded as a modular extension of the calculus \cfik for the logic \fik in \cite{csl-2024-fik}: it is obtained by adding to \cfik an interaction rule capturing $(\bc)$. 

We start by introducing some key syntactic notions. 

\begin{definition}
\label{def:bi-nested:sequent}
  A \textit{bi-nested sequent} (or \emph{sequent} for short) is inductively defined as: 
  \begin{itemize}
    \item 
    the empty sequent $\, \Rightarrow\, $ is a bi-nested sequent;
    \item 
    if $\Gamma$ and $\Delta'$ are multisets of formulas and 
    $S_1,\ldots,S_m$, $T_1,\ldots,T_n$ are bi-nested sequents, for $m,n \geq 0$, then $\seq{\Gamma}{\Delta',\iblock{S_1},\ldots,\iblock{S_m},\mblock{T_1},\ldots,\mblock{T_n}}
    $ is a bi-nested sequent. 
  \end{itemize}
\end{definition}


We use $S,T,U$ to denote sequents. 
Given a sequent $S$, denote the antecedent of $S$ by $\ant{S}$, which is a multiset of formulas, and the succedent of $S$ by $\suc{S}$, which is a multiset of formulas and blocks. The multiset of formulas occurring in $\suc{S}$ is further denoted by $\fm{\suc{S}}$.


\begin{definition}
  The \emph{modal depth} of a formula $A$, denoted as $\textit{md}(A)$, is the maximal nested number of modalities occurring in $A$. For $\Gamma$ finite set of formulas, define $\textit{md}(\Gamma)=\textit{md}(\bigwedge \Gamma)$. 

  For a sequent $S$ of the form 
  $\Gamma\Rightarrow\Delta,[S_1],\ldots,[S_m],\langle T_1\rangle,\ldots,\langle T_n\rangle
  $, 
  let 
  $
  \textit{md}(S)=\max\{\textit{md}(\Gamma),\textit{md}(\Delta),\textit{md}(S_1)+1,\ldots,\textit{md}(S_m)+1,\textit{md}(T_1),\ldots,\textit{md}(T_n)\}
  $. 
\end{definition}

Next, we define the notion of \emph{context}, standard in  deep inference formalisms, which intuitively denotes a sequent with a placeholder. 

\begin{definition}
  A \emph{context} $G\{\ \}$ is 
  inductively defined as:
  \begin{itemize}
    \item The empty context $\{\ \}$ is a context;
    \item If $\Gamma\Rightarrow\Delta$ is a sequent and  $G'\{\ \}$ is a context then both  $\Gamma\Rightarrow\Delta, \langle G'\{\ \}\rangle$ 
    and $\Gamma\Rightarrow\Delta, [G'\{\ \}]$ 
    are contexts. 
  \end{itemize}
\end{definition}

By filling a context $G\{ \ \}$ with some sequent $S$, we obtain a sequent $G\{S\}$. For example, given a context $G\{\ \}=p\wedge q, \square r\Rightarrow \Diamond p,\langle \square p\Rightarrow [\Rightarrow q]\rangle,[\{ \ \}]$ and a sequent $S=p\Rightarrow q\vee r,[r\Rightarrow s]$, we have $G\{S\}=p\wedge q, \square r\Rightarrow \Diamond p,\langle \square p\Rightarrow [\Rightarrow q]\rangle,[p\Rightarrow q\vee r,[r\Rightarrow s]]$. 

Next we introduce an operator on the succedent of a sequent, which 
singles out the formulas occurring in the antecedent of modal blocks occurring deep into the sequent. 

\begin{definition}
  Let $\seq{\Gamma}{\Delta}$ be a sequent. 
  The \emph{local positive part} of $\Delta$, denoted by $\Delta^*$, is defined as: 
  \begin{itemize}
    \item $\Delta^*=\emptyset$ if $\Delta$ is $\mblock{\cdot}$-free;
    \item $\Delta^*=\mblock{\seq{\Lambda_1}{\Theta_1^*}},\ldots,\mblock{\seq{\Lambda_k}{\Theta_k^*}}$ if $\Delta=\Delta_0,\mblock{\seq{\Lambda_1}{\Theta_1}},\ldots,\mblock{\seq{\Lambda_k}{\Theta_k}}$ and $\Delta_0$ is $\mblock{\cdot}$-free.
  \end{itemize}
\end{definition}

\begin{figure}[t]
  \begin{adjustbox}{max width = \textwidth} 
    \begin{tabular}{c c}
    \multicolumn{2}{c}{
      \AxiomC{}
      \RightLabel{\rm ($\bot_L$)}
      \UnaryInfC{$ G\{\Gamma,\bot\seqarrowann{}\Delta\}$}
      \DisplayProof
      \quad 
       \AxiomC{}
      \RightLabel{\rm ($\top_R$)}
      \UnaryInfC{$  G\{\Gamma\seqarrowann{}\top,\Delta\}$}
      \DisplayProof
      \quad
         \AxiomC{}
      \RightLabel{\rm ($\text{id}$)}
      \UnaryInfC{$  G\{\Gamma,p\seqarrowann{}\Delta,p\}$}
      \DisplayProof
    }\\[0.3cm]    
     \AxiomC{$ G\{A,B,\Gamma\seqarrowann{}\Delta\}$}
    \RightLabel{\rm ($\wedge_L$)}
    \UnaryInfC{$ G\{A\wedge B,\Gamma\seqarrowann{}\Delta\}$}
    \DisplayProof
    & 
      \AxiomC{$ G\{\Gamma\seqarrowann{}\Delta ,A\}$}
    \AxiomC{$ G\{\Gamma\seqarrowann{}\Delta,B\}$}
    \RightLabel{\rm ($\wedge_R$)}
    \BinaryInfC{$ G\{\Gamma\seqarrowann{}\Delta,A\wedge B\}$}
    \DisplayProof\\[0.5cm]
     \AxiomC{$  G\{\Gamma,A, \seqarrowann{}\Delta\}$}
    \AxiomC{$ G\{\Gamma,B, \seqarrowann{}\Delta\}$}
    \RightLabel{\rm ($\vee_L$)}
    \BinaryInfC{$ G\{\Gamma,A\vee B\seqarrowann{}\Delta\}$}
    \DisplayProof
    & 
     \AxiomC{$ G\{\Gamma\seqarrowann{} \Delta,A,B\}$}
    \RightLabel{\rm ($\vee_R$)}
    \UnaryInfC{$ G\{\Gamma\seqarrowann{} \Delta,A\vee B\}$}
    \DisplayProof\\[0.5cm]
    \AxiomC{$  G\{\Gamma,A\supset B\seqarrowann{} A,\Delta\}$}
  \AxiomC{$ G\{\Gamma, B\seqarrowann{}\Delta\}$}
  \RightLabel{\rm ($\supset_L$)}
  \BinaryInfC{$ G\{\Gamma,A\supset B\seqarrowann{}\Delta\}$}
  \DisplayProof
   & 
    \AxiomC{$G\{\Gamma\seqarrowann{} \Delta, \langle A\seqarrowann{} B\rangle\}$}
  \RightLabel{\rm ($\supset_R$)}
  \UnaryInfC{$G\{\Gamma\seqarrowann{} \Delta, A\supset B\}$}
  \DisplayProof
 \\[0.5cm]
  \AxiomC{$ G\{\Gamma,\square A\seqarrowann{}\Delta,[\Sigma,A\seqarrowann{} \Pi]\}$}
  \RightLabel{\rm ($\square_L$)}
  \UnaryInfC{$ G\{\Gamma,\square A\seqarrowann{}\Delta,[\Sigma\seqarrowann{} \Pi]\}$}
  \DisplayProof
    & 
      \AxiomC{$ G\{\Gamma\seqarrowann{}\Delta, \langle \, \seqarrowann{} \, [\, \seqarrowann{} A]\rangle\}$}
  \RightLabel{\rm ($\square_R$)}
  \UnaryInfC{$ G\{\Gamma, \seqarrowann{}\Delta,\square A\}$}
  \DisplayProof\\[0.5cm]
     \AxiomC{$  G\{\Gamma\seqarrowann{}\Delta,[A\seqarrowann{}\, ]\}$}
    \RightLabel{\rm ($\Diamond_L$)}
    \UnaryInfC{$ G\{\Gamma,\Diamond A\seqarrowann{}\Delta\}$}
    \DisplayProof
 & 
  \AxiomC{$ G\{\Gamma\seqarrowann{}\Delta,\Diamond A,[\Sigma\seqarrowann{}\Pi,A]\}$}
    \RightLabel{\rm ($\Diamond_R$)}
    \UnaryInfC{$ G\{\Gamma\seqarrowann{}\Delta,\Diamond A,[\Sigma\seqarrowann{}\Pi]\}$}
    \DisplayProof\\[0.5cm]
    \multicolumn{2}{c}{
    \AxiomC{$G \{\Gamma,\Gamma'\seqarrowann{}\Delta,\langle\Gamma',\Sigma\seqarrowann{}\Pi\rangle\}$}
  \RightLabel{\rm (\text{trans})}
  \UnaryInfC{$ G\{\Gamma,\Gamma'\seqarrowann{}\Delta,\langle\Sigma\seqarrowann{}\Pi\rangle\}$}
  \DisplayProof
    }\\[0.5cm]
     \multicolumn{2}{c}{
     \AxiomC{$ G\{\Gamma \seqarrowann{}\Delta,\langle\Sigma\seqarrowann{}\Pi,[\Lambda\seqarrowann{}\Theta^*]\rangle,[\Lambda\seqarrowann{}\Theta]\}$}
  \RightLabel{\rulefc}
  \UnaryInfC{$ G\{\Gamma\seqarrowann{}\Delta,\langle\Sigma\seqarrowann{}\Pi\rangle,[\Lambda\seqarrowann{}\Theta]\}$}
  \DisplayProof
    }\\[0.5cm]
    \multicolumn{2}{c}{
      \AxiomC{
  $G\{\Gamma\seqarrowann{}\Delta,[\Lambda\seqarrowann{}\Theta,\langle\Sigma\seqarrowann{}\Pi\rangle],\langle \, \seqarrowann{}[\aseq{\Sigma}{}{\Pi}]\rangle\}$
  }
  \RightLabel{\rulebc}
  \UnaryInfC{$G\{\Gamma\seqarrowann{}\Delta,[\Lambda\seqarrowann{}\Theta,\langle\Sigma\seqarrowann{}\Pi\rangle]\}$}
  \DisplayProof
    }\\
    \end{tabular}
\end{adjustbox}
  \caption{Rules of \cik}\label{fig:cik}
  \end{figure}

The rules of \cik are presented in Figure~\ref{fig:cik}.  
Propositional rules are standard, with the exception of rule $(\supset_R)$, which from a bottom-up view introduces a new implication block on the premise. 
Intuitively, every block component of a sequent represents a world in a bi-relational model, so the newly created implication block corresponds to an upper $\leq$ world of the sequent/world  $\Gamma \Rightarrow \Delta, A \supset B$. The modal rules either create a new modal block (corresponding to an $R$-successor of existing sequents) or propagate formulas into the existing ones. 
The transfer rule (trans) corresponds to the {hereditary property} of bi-relational models. The two interaction rules \rulefc and \rulebc respectively correspond to forward and backward interaction. Rule \rulefc creates a `copy' of a modal block as successor of an implication block, and transfers in it the local positive information from the original block. Rule \rulebc introduces bottom-up an exact copy of a component, which is `reached' through a different modal and implication path. 

A \emph{derivation} in \cik is a tree of sequents generated by the rules in Figure~\ref{fig:cik}. A \emph{proof} for a sequent $\sq S$ (resp. a formula $A$) in \cik is a derivation having $\sq S$ (resp. $\, \Rightarrow A$) at the root, and whose leaves are instances of $(\bot_L)$, $(\top_R)$ or (id). A sequent or a formula is said to be \emph{provable} in \cik if it has a proof in \cik. 
We can verify that each axiom in $\+H_\ik$ is provable in \cik, for example:



\begin{example}
 Axiom  ($\ax4$) is provable in \cik. We use the following 
 abbreviations: 
 $G_1\{\ \}=\seq{\Diamond p\supset \square q}{\iblock{\{\ \}}}$ and $G_2\{\ \}=\seq{\Diamond p\supset \square q}{[\Rightarrow\langle p\Rightarrow q\rangle],\iblock{\{ \ \}}}$, and the proof is given as 

 \begin{center}
     
\AxiomC{}
\RightLabel{$(\text{id})$}
\UnaryInfC{$G_1\{G_2\{\Diamond p\supset \square q\Rightarrow\Diamond p,[p\Rightarrow q,p]\}\}$}
\RightLabel{$(\Diamond_R)$}
\UnaryInfC{$G_1\{G_2\{\Diamond p\supset \square q\Rightarrow\Diamond p,[p\Rightarrow q]\}\}$}
\AxiomC{}
\RightLabel{$(\text{id})$}
\UnaryInfC{$G_1\{G_2\{\square q\Rightarrow[q,p\Rightarrow q]\}\}$}
\RightLabel{$(\square_L)$}
\UnaryInfC{$G_1\{G_2\{\square q\Rightarrow[p\Rightarrow q]\}\}$}
\RightLabel{$(\supset_L)$}
\BinaryInfC{$G_1\{G_2\{\Diamond p\supset \square q\Rightarrow[p\Rightarrow q]\}\}$}
\RightLabel{\text{(trans)}}
\UnaryInfC{$G_1\{\Diamond p\supset \square q\Rightarrow[\Rightarrow\langle p\Rightarrow q\rangle],\langle\Rightarrow[p\Rightarrow q]\rangle\}$}
\RightLabel{\rulebc}
\UnaryInfC{$G_1\{\Diamond p\supset \square q\Rightarrow[\Rightarrow\langle p\Rightarrow q\rangle]\}$}
\RightLabel{($\supset_R$)}
\UnaryInfC{$\Diamond p\supset \square q\Rightarrow\langle\Diamond p\supset \square q\Rightarrow[\Rightarrow p\supset q]\rangle$}
\RightLabel{\text{(trans)}}
\UnaryInfC{$\Diamond p\supset \square q\Rightarrow\langle\Rightarrow[\Rightarrow p\supset q]\rangle$}
\RightLabel{($\square_R$)}
\UnaryInfC{$\Diamond p\supset \square q\Rightarrow\square(p\supset q)$}
\DisplayProof

 \end{center}
\end{example}

Admissibility of 
various 
weakening and contraction 
rules
can be 
obtained by a standard proof (which we omit) proceeding by induction on the structure of a derivation.

\begin{proposition}\label{prop:admissibility-weakening-contraction}
  The following 
  rules are admissible in \cik,
  
\begin{center}
    \begin{adjustbox}{max width=\textwidth}
      \begin{tabular}{ccccc}
        \AxiomC{$G\{\Gamma\Rightarrow\Delta\}$}
        \RightLabel{$(w_{R})$}
        \UnaryInfC{$G\{\Gamma\Rightarrow\Delta,\+O\}$}
        \DisplayProof
        &
        \quad
        &
        \AxiomC{$G\{\Gamma\Rightarrow\Delta\}$}
        \RightLabel{$(w_{L})$}
        \UnaryInfC{$G\{A,\Gamma\Rightarrow\Delta\}$}
        \DisplayProof
        &
        \quad
        &
        \AxiomC{$\Gamma\Rightarrow\Delta$}
        \RightLabel{$(w_{C})$}
        \UnaryInfC{$G\{\Gamma\Rightarrow\Delta\}$}
        \DisplayProof
        \\
        &&&&\\
        \AxiomC{$G\{\Gamma\Rightarrow\Delta,\+O,\+O\}$}
        \RightLabel{$(c_{R})$}
        \UnaryInfC{$G\{\Gamma\Rightarrow\Delta,\+O\}$}
        \DisplayProof
        &
        \quad
        &
        \AxiomC{$G\{\Gamma,A,A\Rightarrow\Delta\}$}
        \RightLabel{$(c_{L})$}
        \UnaryInfC{$G\{\Gamma,A\Rightarrow\Delta\}$}
        \DisplayProof
        &
        \quad
        &
        \AxiomC{$G\{\Gamma\Rightarrow\Delta,\iblock{\seq{}{}}\}$}
        \RightLabel{$(\iblock{\emptyset})$}
        \UnaryInfC{$G\{\Gamma\Rightarrow\Delta\}$}
        \DisplayProof
      \end{tabular}
    \end{adjustbox}
  \end{center}
  where in $(w_R)$ and $(c_R)$, $\+O$ can be either a formula or a block. 
\end{proposition}


Next, we show the soundness of \cik. Since a bi-nested sequent does not have a formula interpretation, we prove soundness with respect to bi-relational models for \ik. 
We first extend the forcing conditions for formulas to blocks and sequents. 

\begin{definition}
  Let $\mathcal{M}=(W,\leq,R,V)$ be a bi-relational model and $x\in W$. The forcing relation $\Vdash$ is extended to sequents and blocks in the following way:
  \begin{itemize}
    \item $\+M,x \nVdash \ \Rightarrow \ $;
    \item $\+M, x \Vdash [T]$ iff for every $y$ with $Rxy$, 
    it holds that 
    $\+M, y \Vdash T$;
    \item $\+M, x \Vdash \langle T \rangle$ iff for every $x'$ with $x\leq x'$, 
    it holds that 
    $\+M, x' \Vdash T$;
    \item $\+M, x \Vdash \Gamma \Rightarrow \Delta$ iff either $\+M, x \not \Vdash A$ for some $A\in \Gamma$, or 
    $\+M, x \Vdash \+O$ for some $\+O \in \Delta$, where ${\cal O}$ is either a formula or a block. 
  \end{itemize}
  We say that $S$ is \emph{valid in $\+M$} iff for all $ w\in W$, we have $\+M,w\Vdash S$. Moreover, we say that 
  $S$ is \emph{valid} iff it is valid in any bi-relational model.
\end{definition}

Whenever 
$\+M$ is clear from the context, 
we write $x \Vdash \+O$ for an object $\+O$ (a formula, a sequent or a block). 
Moreover, given a sequent $\Gamma \Rightarrow \Delta$, we write $x \Vdash \Delta$ if there is an $\+O\in \Delta$ s.t. $x \Vdash \+O$. Conversely, we write $x \not\Vdash \Delta$ if the condition does not hold. 

\begin{theorem}[Soundness]
\label{thm:soundness}
  If a sequent is provable in \cik, then it is valid.
\end{theorem}

\begin{proof}
  Given a rule ($r$) with empty context, i.e., of the form $\frac{S_1 \quad S_2}{S}$ or $\frac{S_1}{S}$, we say that ($r$) is {\it valid} iff for any model $\+M$ and any world $x$ in it, if $x\Vdash S_i$ for all $i\leq 2$, then $x\Vdash S$. 

  In order to establish soundness, we first need to show that every \cik rule with empty context is valid. Most cases can be found in \cite{csl-2024-fik} and we only prove the case of \rulebc here. 
  Take \rulebc with empty context and, aiming at a contradiction, suppose it is not valid. 
	Then there is a bi-relational model $\+M=(W,\leq,R,V)$ and $x\in W$ such that: 
    \begin{enumerate}[$(a)$]
        \item \label{pr:sound:ita}$x\Vdash \Gamma\Rightarrow\Delta,[\Lambda\Rightarrow\Theta,\langle\Sigma\Rightarrow\Pi\rangle],\langle\Rightarrow[\Sigma\Rightarrow\Pi]\rangle$; and
        \item \label{pr:sound:itb}$x\nVdash \Gamma\Rightarrow\Delta,[\Lambda\Rightarrow\Theta,\langle\Sigma\Rightarrow\Pi\rangle]$.
    \end{enumerate}
    It follows from \ref{pr:sound:itb} that:
    \begin{enumerate}[$(b1)$]
        \item $x\Vdash\Gamma$; and 
        \item $x\nVdash\Delta$; and 
        \item \label{pr:sound:itb3}$x\nVdash[\Lambda\Rightarrow\Theta\langle\Sigma\Rightarrow\Pi\rangle]$.
    \end{enumerate}
    By \ref{pr:sound:itb3} 
	there is some $y$ such that $xRy$ and $y\Vdash\Lambda$, $y\nVdash\Theta$ and $y\nVdash\langle\Sigma\Rightarrow\Pi\rangle$. 
	From the latter  we have that 
	there is a $y'\geq y$ s.t. $y'\nVdash\Sigma\Rightarrow\Pi$.  
  By \ref{pr:sound:ita} and \ref{pr:sound:itb}, we obtain  $x\Vdash \langle\Rightarrow[\Sigma\Rightarrow\Pi]\rangle$, whence 
  for all $z\geq x$, it holds that $ z\Vdash[\Sigma\Rightarrow\Pi]$, i.e., 
  every $R$-successor of such $z$ satisfies $\Sigma\Rightarrow\Pi$. Since $\+M$ satisfies ($\bc$), from $xRy$ and $y\leq y'$ it follows that there is some $z_0$ s.t. $x\leq z_0$ and $z_0Ry'$. Thus $y'\Vdash\Sigma\Rightarrow\Pi$, a contradiction.

 To conclude the proof, we need to show  that each rule ($r$) in \cik
 preserves validity, i.e., for a model $\+M$ and a world $x$ in it, $x\Vdash G\{S_i\}$ implies $x\Vdash G\{S\}$. 
 The proof proceeds by an easy induction on the structure of contexts. 
\end{proof}
\section{Proof search and termination}
\label{sec:proof-serach}


In this section we introduce a decision algorithm for \ik, that is, an algorithm that for any formula decides whether it is valid or not. The algorithm implements a terminating proof search strategy in \cik.

To obtain termination of root-first proof search, we need to prevent loops which would occur by na\"ively applying the rules of \cik bottom-up.  
To this purpose, we adopt 
a cumulative and set-based version of \cik. In this way we avoid redundant applications of the rules, by preventing  several copies of the same formula to be added to a sequent. However, this is not enough, as infinite chains of nested implication blocks might occur in root-first proof search. 
To 
block 
this kind of loop, we $(i)$ modify rule $(\supset_R)$ to avoid unnecessary generation of nested implication blocks, and  $(ii)$ we adopt a \emph{blocking} mechanism that prevents the application of rules to a nested implication block if it is a  `copy' of another implication block occurring lower in the $\iblock{\cdot}$-chain. We notice that both these strategies are \emph{local} to a sequent: 
there is no need to store a whole derivation or branch to decide whether a rule can be applied or not. The controls $(i)$ and $(ii)$, together with a specific order of applications of the rules, allow us to obtain termination of root-first proof search, proved at the end of this section. In Section~\ref{sec:completeness} we will prove the completeness of \cik, by showing how a countermodel for the sequent at the root can be extracted by a `failed' sequent produced by the algorithm.

We start by modifying bi-nested sequents as follows. 
A \emph{set-based bi-nested sequent} (or  \emph{set-sequent} or \emph{sequent}) is defined as in Definition~\ref{def:bi-nested:sequent}, by replacing `multiset' with `set'. 
Accordingly, for a set-based sequent $S =\Gamma \Rightarrow \Delta$ the antecedent $\ant{S}= \Gamma$ is a set of formulas and the consequent $\suc{S} =\Delta$ is a set of formulas and possibly other blocks containing set-based sequents. From now on, we shall employ set-sequents. 

Based on set-sequents, we define a cumulative version of \cik by repeating the principal formula in the premiss(es) of each rule. For instance, rule $(\Box_R)$ becomes: 
\begin{center}
\small
\AxiomC{$G\{\Gamma\seqarrowann{}\Delta,\square A, \langle\seqarrowann{} [\seqarrowann{} A]\rangle\}$}
  \RightLabel{\rm ($\square_R$)}
  \UnaryInfC{$G\{\Gamma, \Diamond A\seqarrowann{}\Delta,\square A\}$}
  \DisplayProof
  \normalsize
\end{center}
To restrict the proof search space, we further replace rule $(\supset_R)$ with  two rules: 
\begin{center}
\small
      \AxiomC{$G\{\aseq{ \Gamma}{}{\Delta},A\supset B,B\}$}
  \LeftLabel{$A\in\Gamma$}
  \RightLabel{\rm ($\supset_R$)}
  \UnaryInfC{$G\{\aseq{ \Gamma}{}{\Delta},A\supset B\}$}
  \DisplayProof
  \quad 
  \AxiomC{$G\{\aseq{\Gamma}{}{\Delta},A\supset B,\langle \aseq{A}{}{B}\rangle\}$}
  \LeftLabel{$A\notin\Gamma$}
  \RightLabel{\rm ($\supset_R$)}
  \UnaryInfC{$G\{\aseq{\Gamma}{}{\Delta},A\supset B\}$}
  \DisplayProof
  \normalsize
\end{center}
We call the resulting proof system \ccik. 
Using Proposition~\ref{prop:admissibility-weakening-contraction},
it is easy to verify that \ccik~is equivalent to \cik (we omit the proof). 

\begin{proposition}
  Let $S$ be a sequent. $S$ is provable in {\rm \cik}~if and only if it is provable in {\rm \textbf{C}\cik}.
\end{proposition}

Next, we define a number of relations, meant to 
capture 
membership and inclusion 
between sequents. In particular, the notion of \emph{structural inclusion} 
will play a crucial role in both the termination strategy and in the definition of the pre-order relation on countermodels, which will be established in Section~\ref{sec:completeness}. 

\begin{definition}
  Let $\sq S_1$, $\sq S_2$ be sequents. 
  Set $\orizero{\sq S_1}{\sq S_2}$ iff $\orseq{\sq S_1} \in \sq S_2 $ and $\rizero{\sq S_1}{\sq S_2}$ iff $\rseq{\sq S_1} \in \sq S_2 $. Then, let 
  $\ori{}{}$ and $\ri{}{}$ be the transitive closure of $\orizero{}{}$ and $\rizero{}{}$ respectively. 
  Then, take $\allizero{}{} = \orizero{}{} \cup \rizero{}{}$ and let $\alli{}{}$ be the reflexive and transitive closure of $\allizero{}{}$. 
\end{definition}

\begin{definition}\label{def:saturation-inclusion}
  Let $\sq S$, $\sq S_1$ and $\sq S_2$ be sequents such that $\alli{\sq S_1}{S}$, $\alli{\sq S_2}{S}$ and $\sq S_1 = \Gamma_1 \Rightarrow \Delta_1$, $\sq S_2 = \Gamma_2 \Rightarrow \Delta_2$. 
  Sequent $\sq S_1$ is \emph{structurally included} in $\sq S_2$, in symbols $\stri{\sq S_1}{\sq S}{\sq S_2}$ iff all of the following hold: 
  \begin{enumerate}[(i)]
      \item $\Gamma_1\subseteq\Gamma_2$; and 
      \item For each $\rizero{ T_1}{ S_1}$, there exists $ \rizero{ T_2}{ S_2}$ s.t. $\stri{\sq T_1}{\sq S}{\sq T_2}$.
    \end{enumerate}
\end{definition}

\begin{figure}
	\begin{adjustbox}{max width=\textwidth}
		
	\begin{tabular}{l l l  l }
		($\wedge_R$) & If $A\wedge B\in \Delta$, then $A\in\Delta$ or $B\in \Delta$.  &($\bot_L$) &  $\bot\notin\Gamma$.\\[0.1cm]
		 ($\wedge_L$) & If $A\wedge B\in \Gamma$, then $A,B\in \Gamma$. 
		 & $(\top_R)$& $\top\notin\Delta$. \\[0.1cm]
		 ($\vee_L$) & If $A\vee B\in \Gamma$, then $A\in\Gamma$ or $B\in\Gamma$. &
		 ($\mbox{id}$) &  $\mathsf{At}\cap(\Gamma\cap\Delta)$ is empty. \\[0.1cm]
		 ($\vee_R$) &If $A\vee B\in\Delta$, then $A,B\in\Delta$. &  &\\[0.1cm]
		 ($\supset_L$)& If $A\supset B\in\Gamma$, then $A\in\Delta$ or $B\in\Gamma$.\\[0.1cm]
		  ($\supset_R$) & 
		   \multicolumn{3}{l}{
		  If $A\supset B\in\Delta$, then either $A\in\Gamma$ and $B\in\Delta$, 
		  or there is $\langle\Sigma\Rightarrow\Pi\rangle\in\Delta$  }\\
		  &  \multicolumn{3}{l}{
		  with $A\in\Sigma$ and $B\in\Pi$.}\\[0.1cm]
	   ($\Diamond_R$) &
	   \multicolumn{3}{l}{ If $\Diamond A\in \Delta$ and $[\Sigma\Rightarrow\Pi]\in\Delta$, then $A\in\Pi$.} \\[0.1cm]
	   ($\Diamond_L$)& 
	   \multicolumn{3}{l}{If $\Diamond A\in\Gamma$, then there is $[\Sigma\Rightarrow\Pi]\in\Delta$ with $A\in\Sigma$. }\\[0.1cm]
	   ($\square_R$)& 
	   \multicolumn{3}{l}{
	    If $\square A\in \Delta$, then either there is $[\Lambda\Rightarrow\Theta]\in\Delta$ with $A\in\Theta$, }\\
    & \multicolumn{3}{l}{ or there is $\langle\Sigma\Rightarrow[\Lambda\Rightarrow\Theta],\Pi\rangle\in\Delta$ with $A\in\Theta$.}\\[0.1cm]
	 ($\square_L$) & 
	 \multicolumn{3}{l}{ If $\square A\in\Gamma$ and $[\Sigma\Rightarrow\Pi]\in\Delta$, then $A\in \Sigma$.}\\[0.1cm]
	 (trans) & If $\langle\Sigma\Rightarrow\Pi\rangle\in\Delta$, then $\Gamma\subseteq\Sigma$. & & \\[0.1cm]
		 \rulefc & 
		 \multicolumn{3}{l}{ If  $\langle\Sigma\Rightarrow\Pi\rangle \in \Delta $ and $[\Lambda\Rightarrow\Theta]\in\Delta$, then there is $[\Phi\Rightarrow\Psi]\in\Pi$ with }\\ 
		 & 
		 \multicolumn{3}{l}{
		 $\Lambda\Rightarrow\Theta\subseteq^S\Phi\Rightarrow\Psi$.
		}\\[0.1cm]
		\rulebc &
		\multicolumn{3}{l}{ If  $[\Sigma\Rightarrow\Pi,\langle S_1\rangle]\in\Delta$, then there is $\langle\Phi\Rightarrow\Psi,[S_2]\rangle\in \Delta$ with $S_1\str S_2$. 
		}\\[0.1cm]
	\end{tabular}
\end{adjustbox}
   \caption{Saturation conditions, \ccik}
\label{fig:saturation:conditions}
\end{figure}

\begin{example}
  Consider the sequents below, all occurring within a sequent $S$: 
  $$
  \begin{array}{l}
    S_1\,  = \, \seq{p,q}{r,\mblock{\seq{s\wedge r}{p,\mblock{\seq{p\vee q}{s,\iblock{\seq{s}{t}}}}}}}\\
    S_2 \,  = \,\seq{p,q, r\supset s}{s,\mblock{\seq{s\wedge r}{\mblock{\seq{p\vee q}{t}}}}}\\
    S_3 \,  = \,\seq{p,q}{r,\iblock{\seq{s\wedge r}{\mblock{\seq{p\vee q}{s,\mblock{\seq{s}{t}}}}}}}
  \end{array}
  $$
  It holds that $\ori{\iblock{\seq{s}{t}}}{S_1} $, $\alli{\iblock{\seq{s}{t}}}{S_1} $, $\alli{S_1}{S_1}$,  
  $S_1\subseteq^S S_2$ and $S_1\not\subseteq^SS_3$. 
\end{example}

Next, we define the \emph{saturation conditions} associated with each rule in \ccik. These conditions are 
necessary 
to prevent redundant application of rules during backward proof search. 

\begin{definition}\label{def:saturation-condition}
  Let $\Gamma\Rightarrow\Delta$ and $S$ be  set-based  sequents with $\alli{\Gamma \Rightarrow \Delta}{S} $. The saturation conditions associated to \ccik rules can be found in Figure~\ref{fig:saturation:conditions}.
\end{definition}

We say a sequent $S$ is \textit{saturated} with respect to a rule $(r)$ if $S$ satisfies the saturation condition associated with $(r)$. 
The following proposition, of which we omit the proof, says $\ibin$-components comply with structural inclusion subject to certain saturation conditions. 

\begin{proposition}\label{prop:sequent-inclusion}
Let $\Gamma\Rightarrow\Delta$ and $S$ be  set-based  sequents with $\alli{\Gamma \Rightarrow \Delta}{S} $ and $\Gamma\Rightarrow\Delta$ saturated with respect to {\rm (trans)} and {\rm \rulefc}. 
If $\Delta$ is of the form $\Delta',\langle\Sigma\Rightarrow\Pi\rangle$, then $\Gamma\Rightarrow\Delta\str\Sigma\Rightarrow\Pi$.
\end{proposition}

In 
backward proof search, we say $(r)$ is applied \textit{redundantly} to $S$ if $S$ is already saturated with $(r)$. We set two basic constraints for each rule application in root-first proof search, namely: 
\textit{no rule is applied to an axiom}; and 
\textit{no rule is applied redundantly}.
However, a proof-search strategy implementing only 
these two 
basic constraints 
is not sufficient to 
ensure termination of proof search. 
Similarly as in~\cite{csl-2024-fik}, loops in proof search can be generated by unrestricted applications of $(\square_R)$ and $(\supset_R)$. The following example is 
taken from~\cite{csl-2024-fik}.

\begin{example}
\label{ex:non:termination}
  Consider a derivation having sequent 
  $S = \Box a\supset\bot, \Box b\supset \bot  \Rightarrow $ at the root. 
  Let $\Gamma = \Box a\supset\bot, \Box b\supset \bot $. Below is depicted an infinite branch generated by backward proof-search. The numbers next to the rules denote multiple applications of the rules.
  \begin{center}
    \small
	\AxiomC{ 
	$
	\begin{matrix}
	\vdots\\
	\Gamma \Rightarrow \Box a, \Box b, \langle \Gamma \Rightarrow \Box a, \Box b, [\Rightarrow a], \langle \Gamma \Rightarrow \Box a, \Box b, [\Rightarrow b]\rangle \rangle, \langle \Gamma \Rightarrow \Box a, \Box b, [\Rightarrow b]\rangle
	\end{matrix}
	$
	}
	\UnaryInfC{$\vdots$}
	\UnaryInfC{$\Gamma \Rightarrow \Box a, \Box b, \langle \Gamma \Rightarrow \Box a, \Box b, [\Rightarrow a], \langle \Rightarrow [\Rightarrow b]\rangle \rangle, \langle \Gamma \Rightarrow \Box a, \Box b, [\Rightarrow b]\rangle$}
	\RightLabel{$(\Box_R)$}
	\UnaryInfC{$\Gamma \Rightarrow \Box a, \Box b, \langle \Gamma \Rightarrow \Box a, \Box b,  [\Rightarrow a]\rangle, \langle \Gamma \Rightarrow \Box a, \Box b,  [\Rightarrow b]\rangle$}
	\RightLabel{$(\supset_L)^4$}
	\UnaryInfC{$\Gamma \Rightarrow \Box a, \Box b, \langle \Gamma \Rightarrow [\Rightarrow a]\rangle, \langle \Gamma\Rightarrow [\Rightarrow b]\rangle$}
	\RightLabel{$(\Box_R)^2,(\text{trans})^2$}
	\UnaryInfC{$\Gamma \Rightarrow\Box a, \Box b$}
	\RightLabel{$(\supset_L)^2$}
	\UnaryInfC{$\Gamma \Rightarrow$}
	\DisplayProof
  \end{center}
What happens is that rules applied to
$\square a$ and $\square b$ alternate creating new nested implication blocks; while saturation conditions for the two kinds of blocks cannot be simultaneously satisfied by one single implication block. 
\end{example}

To prevent this kind of loop, 
we make use of the same blocking mechanism 
introduced 
in~\cite{csl-2024-fik} which prevents application of the rules to 
nested implication blocks which are 
`copies' of `previous' sequents/implication blocks. 
We introduce some definitions to formalise the notion of `copy of a previous block'. The  $\sharp$-operator below 
separates the `local' part of the consequent of a sequent.

\begin{definition}
  Let $\seq{\Gamma}{\Delta}$ be a sequent. We define 
  $\Delta^\sharp$ inductively as follows:
  \begin{itemize}
    \item $\Delta^\sharp=\Delta$ if $\Delta$ is block-free;
    \item $\Delta^\sharp=\Delta',\mblock{\seq{\Lambda_1}{\Theta_1^\sharp}},\ldots,\mblock{\seq{\Lambda_n}{\Theta_n^\sharp}}$ if $\Delta=\Delta',\iblock{\seq{\Sigma_1}{\Pi_1}},\ldots,\iblock{\seq{\Sigma_m}{\Pi_m}},\mblock{\seq{\Lambda_1}{\Theta_1}},\ldots,\mblock{\seq{\Lambda_n}{\Theta_n}}$ and $\Delta'$ is block-free.
  \end{itemize}
\end{definition}

For a set-based sequent $\seq{\Gamma}{\Delta}$, 
$\Delta^\sharp$ is just a set of formulas and modal blocks. Using the definition above, we further define the notion of $\sharp$-equivalence. 

\begin{definition}
\label{def:block-eq}
  Let $S_1= \Gamma_1\Rightarrow\Delta_1$, $S_2= \Gamma_2\Rightarrow\Delta_2$ be set-sequents. 
  We say that $S_1$ is \emph{$\sharp$-equivalent} to $S_2$, denoted as $S_1\simeq S_2$, if $\Gamma_1=\Gamma_2$ and $\Delta_1^\sharp=\Delta_2^\sharp$.
\end{definition}

\begin{example}
  Consider $S_1=\seq{a}{b,\iblock{\seq{c}{d,\mblock{\seq{e}{f}}}},\mblock{\seq{g}{h}}}$ and $S_2=\seq{a}{b,\mblock{\seq{g}{h,\iblock{\seq{c}{d}}}},\iblock{\seq{e}{f}}}$. 
  It holds that $Con(S_1)^\sharp=Con(S_2)^\sharp=b,\mblock{\seq{g}{h}}$. Also, $Ant(S_1)=Ant(S_2)$, hence $S_1\simeq S_2$.
\end{example}

Next, we divide \ccik rules of into four groups and define corresponding saturation levels. This is required to specify an order of application of the rules. 
\begin{enumerate}
  \item[(R1)] Basic rules: all propositional and modal rules except $(\square_R)$ and $(\supset_R)$;
  \item[(R2)] Rules transferring formulas or blocks: $(\text{trans})$ and \rulefc;
  \item[(R3)] Rules creating implication blocks from formulas: $(\square_R),(\supset_R)$;
  \item[(R4)] Rules creating implication blocks from blocks: \rulebc.
\end{enumerate}

\begin{definition}\label{3-saturation}
  Let $S =\Gamma\Rightarrow\Delta$ be a non-axiomatic sequent. $S$ is said to be: 
  \begin{itemize}
    \item {\it R1-saturated} if each $T\in^+\Gamma\Rightarrow\Delta^\sharp$ satisfies all the saturation conditions of the R1 rules;
    \item {\it R2-saturated} if $S$ is R1-saturated and $S$ satisfies saturation conditions of the R2 rules for blocks $\langle S_1\rangle,[S_2]$ s.t. $S_1 \in^{\langle\cdot\rangle}_0 S$ and $S_2 \in^{[\cdot]}_0 S$;
    \item {\it R3-saturated} if $S$ is R2-saturated and $S$ satisfies saturation conditions of the R3 rules for each $\square A, A\supset B\in \Delta$;
    \item {\it R4-saturated} if $S$ is R3-saturated and $S$ satisfies saturation conditions of the R4 rule for each $\seq{\Sigma}{\Pi,\iblock{U}}\mbin S$.
  \end{itemize}
\end{definition}

We can now state our blocking condition:

\begin{definition}\label{def:blocked-sequent}
  Let $S$ be a sequent and $S_1, S_2 \seqin S$ where $S_1=\Gamma_1\Rightarrow\Delta_1$ and  $S_2=\Gamma_2\Rightarrow\Delta_2$. 
  We say \emph{$S_2$ is blocked by $S_1$ in $S$} if $S_1$ is R3-saturated, $S_2\in^{\langle\cdot\rangle}S_1$ and $S_1\simeq S_2$. 
  We say that a sequent \emph{$T$ is blocked in $S$} if there exists $T'\seqin S$ such that $T$ is blocked by $T'$ in $S$.
\end{definition}

\begin{example}
Referring to Example~\ref{ex:non:termination}, if we were to continue applying rules $(\supset_L)$, $(\square_R)$ and $(\text{trans})$ to the topmost sequent, we would obtain the sequent $S= \Gamma \Rightarrow  \Box a, \Box b, 
 \langle T_1 \rangle, \langle \Gamma \Rightarrow \Box a, \Box b, [\Rightarrow b]\rangle$, where:\\\vspace{-0.5cm}
 $$
 T_1 \quad = \quad\Gamma \Rightarrow \Box a, \Box b, [\Rightarrow a], \langle \Gamma \Rightarrow \Box a, \Box b, [\Rightarrow b],\langle 
 \overbrace{
 \Gamma \Rightarrow \Box a, \Box b, [\Rightarrow a]}^{T_2}
 \rangle 
 \rangle
 $$
It holds that $T_1 \simeq T_2$ and $T_1$ is R3-saturated, whence $T_2$ is blocked by $T_1$ in $S$. 
Rule applications on the other implication block occurring in $S$, that is, $\iblock{\seq{\Gamma}{\square a,\square b,\mblock{\seq{}{b}}}}$, 
 will be blocked in a similar way. 
\end{example}

%

We further need to extend 
the levels of  saturation in Definition~\ref{3-saturation} by taking into account blocked sequents and by considering \emph{global} saturation, that is, saturation not only of sequents but also of their nested components.

\begin{definition}
  Let $S =\Gamma\Rightarrow\Delta$ be a non-axiomatic sequent. For $i\in\{1,2,3\}$, $S$ is 
 {\it global-R$i$-saturated} iff for each $T\in^+S$, $T$ is either R$i$-saturated or blocked;  $S$ is {\it global-saturated} iff for each $T\in^+S$, $T$ is either R4-saturated or blocked.
\end{definition}

In order to specify the proof-search algorithm, we make use of the following four terminating  macro-procedures 
that extend a given derivation $\+D$ by applying rules to a leaf $S$. 
Each procedure applies rules {\em non-redundantly} to some $T=\Gamma\Rightarrow\Delta\in^+S$, which equivalently means that $S = G\{T\}$ for some context $G$.
\begin{itemize}
  \item $\textbf{EXP1}(\+D, S, T)= \+D'$ where $\+D'$ is the extension of $\+D$  obtained by applying the R1 rules to every formula in $\Gamma\Rightarrow\Delta^\sharp$. 
  \item $\textbf{EXP2}(\+D, S, T)= \+D'$ where $\+D'$ is the extension of $\+D$  obtained by applying the R2 rules to pairs of blocks $\langle T_i\rangle,[T_j]\in\Delta$. 
  \item $\textbf{EXP3}(\+D, S, T)= \+D'$ where $\+D'$ is the extension of $\+D$ obtained by applying the R3 rules to each $\square A,A\supset B\in\Delta$.
  \item $\textbf{EXP4}(\+D, S, T)= \+D'$ where $\+D'$ is the extension of $\+D$ obtained by applying the R4 rule to each $\seq{\Sigma}{\Pi,\iblock{U}}\mbin_0 T$.
\end{itemize}


\begin{proposition}\label{prop:exp123term}
	Let $\+D$ be a finite derivation, $S$ a finite leaf of $\+D$ and $T\in^+S$. Then for $i\in\{1,2,3,4\}$, each $\textbf{EXPi}(\+D, S, T)$ terminates by producing a finite expansion of $\+D$ where all sequents in it are finite. 
\end{proposition}
\begin{proof}
   For \textbf{EXP2}, \textbf{EXP3} and \textbf{EXP4},
the proof is straightforward, 
as each procedure applies rules only on the top level of the involved sequent which is finite. For \textbf{EXP1}, which applies rules also deep on nested components, the proof is less straightforward and can be found in~\cite{csl-2024-fik}.
\end{proof}

Algorithm~\ref{PROCEDURE} showcases the procedure \proc, taking in input the formula $A$ to for which we wish to test derivability.  
\proc implements a breadth-first procedure, returning either a proof for $\Rightarrow A$ or a finite derivation where all the leaves are global-saturated. 
The steps in lines $7 - 17$ operate in parallel on non-axiomatic leaves, trying to make them global-saturated. 
On line 10, the condition says that when applying \rulebc, 
for each $\seq{\Sigma}{\Pi,\iblock{U}}\mbin_0 T$, sequent $U$ needs to be R4-saturated already, thus ensuring that  \rulebc is applied to the \emph{innermost} blocks first. 




\begin{algorithm}[!t]
  \KwIn{$\+D =\ \Rightarrow A$}
    \Repeat{FALSE}{
      \uIf{all the leaves of $\+D$ are axiomatic}{
        return ``PROVABLE" and $\+D$ }
      \uElseIf{all the non-axiomatic leaves of $\+D$ are global-saturated}{return ``UNPROVABLE" and $\+D$} 
    \Else{
        { 
        {\bf select} one non-axiomatic leaf $S$ of $\+D$ that is not global-saturated\\
        \Repeat{FALSE}{
        {
        \uIf{$S$ is global-R3-saturated}{
          for all non-R4-saturated $T\in^+ S$ s.t. for each $\seq{\Sigma}{\Pi,\iblock{U}}\mbin_0 T$, $U$ is R4-saturated, 
        let $\+D = \textbf{EXP4}(\+D, S, T)$}
        \uElseIf{$S$ is global-R2-saturated}{for all non-R3-saturated $T\in^+S$ such that $T$ is $\in^{\langle\cdot\rangle}$-minimal, check whether $T$ is blocked in $S$, if not, let $\+D = \textbf{EXP3}(\+D, S, T)$}
        \uElseIf{$S$ is global-R1-saturated}
            {for all non-R2-saturated $T\in^+ S$, let $\+D = \textbf{EXP2}(\+D, S, T)$}
         \uElse{for all non-R1-saturated $T\in^+ S$, let $\+D = \textbf{EXP1}(\+D, S, T)$}
        }
        }
         }
        }
    }
    \caption{PROCEDURE($A$)}\label{PROCEDURE}
\end{algorithm}



The following lemma, 
whose proof is easy and therefore omitted,  
guarantees 
that the properties of being saturated / blocked  are preserved throughout iterations of the \textbf{repeat} loop (lines $2-17$) of \proc.

\begin{lemma}[Invariance]\label{lemma:invariant}
  Let $\+D$ be a derivation rooted by $\Rightarrow A$, and let $S$ be a leaf of $\+D$. It holds that: 
  \begin{enumerate}
    \item For $T\in^+ S$ where $T = \Gamma \Rightarrow \Delta$ and for any rule {\rm (r)}, 
    if $T$ satisfies the {\rm (r)}-saturation condition on some formulas $A_i$ and/or blocks $\langle T_j\rangle, [T_k]$ {\em before} the execution of (the body of) the \emph{\textbf{repeat}} loop, then $T$ is  saturated w.r.t. the {\rm (r)}-condition on the involved $A_i,\langle T_j\rangle, [T_k]$ {\em after} the execution.
    \item For $T\in^+ S$, if $T$ is blocked in $S$ {\em before} the execution of (the body of) the \emph{\textbf{repeat}}  loop, then it remains blocked {\em after} it.
  \end{enumerate}
\end{lemma}

It remains to show that \proc terminates in a finite number of steps. The proof of the following can be found in~\cite{csl-2024-fik}.

\begin{lemma}\label{lem:eq-class-finite}
  Let $A$ be a formula and $\textbf{Seq}(A)$ the set of sequents that may occur in any possible derivation rooted by $\Rightarrow A$. Denote the quotient of $\textbf{Seq}(A)$ with respect to the $\sharp$-equivalence $\simeq$ by $\textbf{Seq}(A)/_{\simeq}$. Then $\textbf{Seq}(A)/_{\simeq}$ is finite.
\end{lemma}

Since we check the blocking before applying $\textbf{EXP3}(\cdot)$ and we prevent R3-applications to these blocked sequents, we can prove the following (see~\cite{csl-2024-fik}). 

\begin{proposition}\label{coro:no-inf-chain}
  Let $A$ be a formula and $S$ a sequent occurring in a derivation produced by \proc. Then there is no infinite $\ibin$-chain involving infinitely-many $\mathbf{EXP3}(\cdot)$ phases in $S$.
\end{proposition}

With the following proposition, termination follows as an immediate corollary. 

\begin{proposition}\label{lem:exp4-fin}
  Let $A$ be a formula, there are only finitely-many $\mathbf{EXP4}(\cdot)$ phases which can be executed in \proc. 
\end{proposition}

\newcommand{\degree}[1]{deg_{\bc}(#1)}

\begin{proof}
  Assume for the sake of a contradiction that there are infinitely many $\textbf{EXP4}(\cdot)$ phases executed in \proc. 
  Then there must be an infinite branch $\+B = (S_i)_{i \in \omega}$ in a derivation $\+D$ produced by \proc containing infinitely many $\textbf{EXP4}(\cdot)$ phases. 
  By Proposition~\ref{coro:no-inf-chain}, $\+B$  contains only finitely many R3 steps. Thus there is $k\geq 1$ such that for all $j> k$, no R3 step is applied to $S_j$. Our goal is to prove that only finitely many applications of \rulebc~are possible to any arbitrary $S_j$ with $j>k$, whence only finitely many $\textbf{EXP4}(\cdot)$ phases are executable on $S_j$. 
  Using Proposition \ref{coro:no-inf-chain}, which establishes finiteness of $\textbf{EXP1}(\cdot)$ and $\textbf{EXP2}(\cdot)$, we can conclude that the sub-branch of $\+B$ starting with $S_j$ is finite, contradicting with the hypothesis. 
  
  To this purpose we assign to $S_j$ a {\it R4-degree}, defined as the maximum of the modal depths of  nested sequents $T\in^+ S_j$ such that \rulebc~can be applied to $T$ according to the algorithm (and saturation). Formally: 
  \vspace{-0.2cm}
  $$
  \degree{S_j}
  = max\{\mbox{modal depth of} \ T\in^+ S_j \mid \mbox{\rulebc~can be applied to }T\}.
  $$
  We prove by induction on $ \degree{S_j}$ that there can be only finitely many application of \rulebc~to $S_j$. 
  Suppose that  $  \degree{S_j}=1$, and let $S_j = G\{T\}$,  for some context  $G\{\ \}$ and $T = \seq{\Gamma}{\Delta,\mblock{\seq{\Sigma}{\Pi,\iblock{\seq{\Lambda}{\Theta}}}}} $. 
  Suppose that \rulebc~is applicable to $T_1 = \seq{\Sigma}{\Pi,\iblock{\seq{\Lambda}{\Theta}}}$. 
  Since $  \degree{S_j}=1$, $T$ is not contained in any modal block (but it might be contained in an implication block). We have that $S_{j+1} = G\{ \seq{\Gamma}{\Delta,\mblock{\seq{\Sigma}{\Pi,\iblock{\seq{\Lambda}{\Theta}}}} \iblock{\seq{}{\mblock{\seq{\Lambda}{\Theta}}}} } \}$.
  
  Notice that by hypothesis $\seq{\Lambda}{\Theta}$ is R4 saturated and $S_j$ is global-R3-saturated (whence also $T$). Further expansion by $\textbf{EXP1}(\cdot)$ and $\textbf{EXP2}(\cdot)$ to the newly created block leads to a sequent  $S_{t} = G\{ \seq{\Gamma}{\Delta,\mblock{\seq{\Sigma}{\Pi,\iblock{\seq{\Lambda}{\Theta}}}} \iblock{\seq{\Gamma'}{\Delta',\mblock{\seq{\Lambda'}{\Theta'}}}} } \}$, for $t \geq j$.
 
 Observe that $\Delta'$ cannot contain any implication block (even nested), and  $\seq{\Lambda'}{\Theta'}$ cannot contain any implication block that is not occurring in $\seq{\Lambda}{\Theta}$, which by hypothesis is R4-saturated. Thus no \rulebc is applicable to $\seq{\Gamma'}{\Delta',\mblock{\seq{\Lambda'}{\Theta'}}}$. Since $S_j$ contains only finitely-many $T_1$ as above to which \rulebc \ can be applied, we obtain that the sub-branch of $\+B$ starting with $S_j$ can contain only finitely many $\textbf{EXP4}(\cdot)$ phases. 
 
 Suppose now that $\degree{S_j} = s > 1$. Let us consider an \emph{innermost} application of \rulebc~in $S_j = G\{ \seq{\Gamma}{\Delta,\mblock{\seq{\Sigma}{\Pi,\iblock{\seq{\Lambda}{\Theta}}}}}  \}$ for some context $G\{\ \}$. 
 Further assume that the depth of 
  $T = \seq{\Gamma}{\Delta,\mblock{\seq{\Sigma}{\Pi,\iblock{\seq{\Lambda}{\Theta}}}}}$  is $s> 0$.
 We have that $S_{j+1} = G\{\seq{\Gamma}{\Delta,\mblock{\seq{\Sigma}{\Pi,\iblock{\seq{\Lambda}{\Theta}}}},\iblock{\seq{}{\mblock{\seq{\Lambda}{\Theta}}}}} \}$. 
 Sequent $S_{j+1}$ can be further expanded  by applying $\textbf{EXP1}(\cdot)$ and $\textbf{EXP2}(\cdot)$ to the new block, obtaining $S_{t} = G\{\seq{\Gamma}{\Delta,\mblock{\seq{\Sigma}{\Pi,\iblock{\seq{\Lambda}{\Theta}}}},\iblock{\seq{\Gamma'}{\Delta',\mblock{\seq{\Lambda'}{\Theta'}}}}}   \}$.
 As in the previous case, no \rulebc~applications are possible on the block $\seq{\Gamma'}{\Delta',\mblock{\seq{\Lambda'}{\Theta'}}}$. The only possible applications of \rulebc~involving the newly created block are on a modal block that contains $T$. 
 If there are no such  blocks the claim is proved, as there are only finitely many $T\in^+S_j$ to which \rulebc~is applicable. Otherwise suppose that there is a modal block containing $T$, thus we have 
 $S_{j} = G'\{\seq{\Phi}{\Psi, 
 	\mblock{
 		\seq{\Gamma}{\Delta,\mblock{\seq{\Sigma}{\Pi,\iblock{\seq{\Lambda}{\Theta}}}} } } }  \}$. 
It follows that there is a sequent $S_{t} = G'\{\seq{\Phi}{\Psi, 
 	\mblock{
 \seq{\Gamma}{\Delta,\mblock{\seq{\Sigma}{\Pi,\iblock{\seq{\Lambda}{\Theta}}}},\iblock{\seq{\Gamma'}{\Delta',\mblock{\seq{\Lambda'}{\Theta'}}}}} } }  \}$, for $t \geq j$. 
Notice that the application of \rulebc~to $\seq{\Phi}{\Psi, 
	\mblock{
		\seq{\Gamma}{\Delta,\mblock{\seq{\Sigma}{\Pi,\iblock{\seq{\Lambda}{\Theta}}}},\iblock{\seq{\Gamma'}{\Delta',\mblock{\seq{\Lambda'}{\Theta'}}}}} } }$
has now a depth $< s$. 
For  $u > t$, let $S_u$ be the sequent obtained by  finitely many  applications of \rulebc~to all maximal  $T\in^+S_t$ has $d_{bc}(S_u) < s$. We  conclude the proof by applying the induction hypothesis on $S_u$.
%
\end{proof}

\begin{theorem}[Termination]\label{termination}
  Let $A$ be a formula. Proof-search for the sequent $\Rightarrow A$ terminates and it yields either a proof of $A$ or a finite derivation with at least one global-saturated leaf.
\end{theorem}
\section{Completeness}
\label{sec:completeness}

The procedure \proc 
introduced in Section~\ref{sec:proof-serach} terminates in a finite number of steps, producing either a proof of the formula $A$ at the root or a finite derivation where there is one global-saturated leaf. In this section we show how to construct a countermodel for an unprovable formula 
from a global-saturated sequent, thus establishing the completeness of \cik. 

In order to extract a model from a global-saturated sequent, however, we need to additionally keep track of specific blocks occurring in set-based sequents. To this aim, we shall define an annotated version of \ccik, called \ccikann, which operates on sequents whose components are decorated by natural numbers, the \emph{annotations}, and are equipped with an additional structure to keep in memory a binary relation on such annotations. The use of annotations plays no role in our algorithm: thus, \proc can be applied to \ccikann, ignoring the annotation and the additional structure. For ease of exposition, we choose not to present \ccikann in the previous section, and introduce annotations here, as they are only needed for the countermodel construction.



\begin{definition}
  An \emph{annotated sequent} is a set-sequent $S$ where the sequent arrow of every component of $S$ is decorated by a $n \in \mathbb{N}$, its \emph{annotation}. 
  An \emph{enriched sequent} is
  a pair $ \+A; \sq S$ where $\sq S$ is an annotated sequent, and $\+A \subseteq \mathbb{N}\times\mathbb{N}$.  
\end{definition}

\begin{definition}\label{def:fresh-copy}
  Let $ \+A; \sq S$ be an enriched sequent and $T\seqin S$. A \emph{fresh annotated copy} of $T$, denoted by $T^\circlearrowleft$, is the annotated sequent obtained by replacing the annotations $\ell_1,..,\ell_n$ of each $U_1,..,U_n\seqin T$ with pairwise distinct annotations $\ell_1',..,\ell_n'$,  such that $\ell_1',..,\ell_n'$ are not annotations in $S$.
\end{definition}
 
\begin{example}
    The following is an enriched sequent: $(3, 5);{A, B} \seqarrowann{3}{C , \rseq{{D}\seqarrowann{5}{E}}}$. Moreover,  $(\aseq{D}{5}{E})^\circlearrowleft$ can be any $\aseq{D}{\ell}{E}$ where $\ell\notin\{3,5\}$, e.g., $\aseq{D}{6}{E}$. 
\end{example}


We can now define \ccikann, a proof system operating on enriched sequents. 
The rules of \ccikann are the rules of \ccik, where every set-sequent is replaced by an enriched sequent and, whenever a rule introduces a new block in its premiss, the block is annotated with a fresh annotation. So, for instance, rule ($\Box_R$) becomes the following: 
\begin{center}
\small
    \AxiomC{$\stack G\{\Gamma\seqarrowann{n}\Delta,\square A, \langle\, \seqarrowann{k} [\, \seqarrowann{h} A]\rangle\}$}
  \RightLabel{\rm ($\square_R$)}
  \UnaryInfC{$\stack G\{\Gamma, \seqarrowann{n}\Delta,\square A\}$}
  \DisplayProof
\end{center}
where $k$ and $h$ are not used as annotation in the conclusion of the rule. The additional structure $\+A$ only plays a role in rule  \rulebc, where $k$ is not used as annotation in the conclusion: 
\begin{center}
\small
     \AxiomC{
  $\+A';~  G\{\Gamma\seqarrowann{n}\Delta,[\Lambda\seqarrowann{m}\Theta,\langle\Sigma\seqarrowann{\ell}\Pi\rangle],\langle \, \seqarrowann{k}[(\aseq{\Sigma}{\ell}{\Pi})^\circlearrowleft]\rangle\}$
  }
  \RightLabel{\rulebc}
  \UnaryInfC{$\stack G\{\Gamma\seqarrowann{n}\Delta,[\Lambda\seqarrowann{m}\Theta,\langle\Sigma\seqarrowann{\ell}\Pi\rangle]\}$}
  \DisplayProof
  \normalsize
\end{center}
for
$
\+A' = \+A\cup\{(j,j')~|~\aseq{\Phi}{j}{\Psi}\in^+\aseq{\Sigma}{\ell}{\Pi}~\text{and}~\aseq{\Phi}{j'}{\Psi}\in^+(\aseq{\Sigma}{\ell}{\Pi})^\circlearrowleft\}
$. The full proof system \ccikann is given in Figure \ref{fig:cik:set:annotated}. A sequent $\Rightarrow A$ is derivable in \ccik iff $\emptyset; \seqarrowann{0} A$ is derivable in \ccikann, and all notions 
introduced in Section~\ref{sec:proof-serach} can be extended to enriched sequents. 

\begin{figure}[!t]
  \begin{adjustbox}{max width = \textwidth} 
  \small{
    \begin{tabular}{c @{}c}
    \multicolumn{2}{c}{
      \AxiomC{}
      \RightLabel{\rm ($\bot_L$)}
      \UnaryInfC{$\stack G\{\Gamma,\bot\seqarrowann{n}\Delta\}$}
      \DisplayProof
      \quad 
       \AxiomC{}
      \RightLabel{\rm ($\top_R$)}
      \UnaryInfC{$ \stack G\{\Gamma\seqarrowann{n}\top,\Delta\}$}
      \DisplayProof
      \quad
         \AxiomC{}
      \RightLabel{\rm ($\text{id}$)}
      \UnaryInfC{$ \stack G\{\Gamma,p\seqarrowann{n}\Delta,p\}$}
      \DisplayProof
    }\\[0.3cm] 
     \AxiomC{$\stack G\{A,B,A\wedge B,\Gamma\seqarrowann{n}\Delta\}$}
    \RightLabel{\rm ($\wedge_L$)}
    \UnaryInfC{$\stack G\{A\wedge B,\Gamma\seqarrowann{n}\Delta\}$}
    \DisplayProof
    & 
     \AxiomC{$\stack G\{\Gamma\seqarrowann{n} \Delta,A\vee B,A,B\}$}
    \RightLabel{\rm ($\vee_R$)}
    \UnaryInfC{$\stack G\{\Gamma\seqarrowann{n} \Delta,A\vee B\}$}
    \DisplayProof
    \\[0.5cm]
    \multicolumn{2}{c}{
    \AxiomC{$\stack G\{\Gamma\seqarrowann{n}\Delta, A\wedge B,A\}$}
    \AxiomC{$\stack G\{\Gamma\seqarrowann{n}\Delta,A\wedge B,B\}$}
    \RightLabel{\rm ($\wedge_R$)}
    \BinaryInfC{$\stack G\{\Gamma\seqarrowann{n}\Delta,A\wedge B\}$}
    \DisplayProof
    }
    \\[0.5cm]
   \multicolumn{2}{c}{
     \AxiomC{$ \stack G\{\Gamma,A, A\vee B\seqarrowann{n}\Delta\}$}
    \AxiomC{$\stack G\{\Gamma,B, A\vee B\seqarrowann{n}\Delta\}$}
    \RightLabel{\rm ($\vee_L$)}
    \BinaryInfC{$\stack G\{\Gamma,A\vee B\seqarrowann{n}\Delta\}$}
    \DisplayProof
    }
    \\[0.5cm]
    \multicolumn{2}{c}{
    \AxiomC{$ \stack G\{\Gamma,A\supset B\seqarrowann{n} A,\Delta\}$}
  \AxiomC{$\stack G\{\Gamma, A\supset B,B\seqarrowann{n}\Delta\}$}
  \RightLabel{\rm ($\supset_L$)}
  \BinaryInfC{$\stack G\{\Gamma,A\supset B\seqarrowann{n}\Delta\}$}
  \DisplayProof
  }
 \\[0.5cm]
  \AxiomC{$ \stack G\{\aseq{\Gamma}{n}{\Delta},A\supset B,B\}$}
  \LeftLabel{$A\in\Gamma$}
  \UnaryInfC{$\stack G\{\aseq{ \Gamma}{n}{\Delta},A\supset B\}$}
  \DisplayProof
  &
  \AxiomC{$\stack G\{\aseq{\Gamma}{n}{\Delta},A\supset B,\langle \aseq{A}{k}{B}\rangle\}$}
  \LeftLabel{$A\notin\Gamma, \star$}
  \UnaryInfC{$\stack G\{\aseq{\Gamma}{n}{\Delta},A\supset B\}$}
  \DisplayProof
\\[0.5cm]
  \AxiomC{$\stack G\{\Gamma,\square A\seqarrowann{n}\Delta,[\Sigma,A\seqarrowann{m} \Pi]\}$}
  \RightLabel{\rm ($\square_L$)}
  \UnaryInfC{$\stack G\{\Gamma,\square A\seqarrowann{n}\Delta,[\Sigma\seqarrowann{m} \Pi]\}$}
  \DisplayProof
    & 
      \AxiomC{$\stack G\{\Gamma\seqarrowann{n}\Delta,\square A, \langle\seqarrowann{k} [\seqarrowann{h} A]\rangle\}$}
      \LeftLabel{$\star$}
  \RightLabel{\rm ($\square_R$)}
  \UnaryInfC{$\stack G\{\Gamma\seqarrowann{n}\Delta,\square A\}$}
  \DisplayProof\\[0.5cm]
     \AxiomC{$ \stack G\{\Gamma\seqarrowann{n}\Delta,[A\seqarrowann{k}]\}$}
    \RightLabel{\rm ($\Diamond_L$)}
    \UnaryInfC{$\stack G\{\Gamma,\Diamond A\seqarrowann{n}\Delta\}$}
    \DisplayProof
 & 
  \AxiomC{$\stack G\{\Gamma\seqarrowann{n}\Delta,\Diamond A,[\Sigma\seqarrowann{m}\Pi,A]\}$}
    \RightLabel{\rm ($\Diamond_R$)}
    \UnaryInfC{$\stack G\{\Gamma\seqarrowann{n}\Delta,\Diamond A,[\Sigma\seqarrowann{m}\Pi]\}$}
    \DisplayProof\\[0.5cm]
    \multicolumn{2}{c}{
    \AxiomC{$\stack G \{\Gamma,\Gamma'\seqarrowann{n}\Delta,\langle\Gamma',\Sigma\seqarrowann{m}\Pi\rangle\}$}
  \RightLabel{\rm (\text{trans})}
  \UnaryInfC{$\stack G\{\Gamma,\Gamma'\seqarrowann{n}\Delta,\langle\Sigma\seqarrowann{m}\Pi\rangle\}$}
  \DisplayProof
    }\\[0.5cm]
    \multicolumn{2}{c}{
     \AxiomC{$\stack G\{\Gamma \seqarrowann{n}\Delta,\langle\Sigma\seqarrowann{m}\Pi,[(\aseq{\Lambda}{\ell}{\Theta^*})^\circlearrowleft]\rangle,[\Lambda\seqarrowann{\ell}\Theta]\}$}
  \RightLabel{\rulefc}
  \UnaryInfC{$\stack G\{\Gamma\seqarrowann{n}\Delta,\langle\Sigma\seqarrowann{m}\Pi\rangle,[\Lambda\seqarrowann{\ell}\Theta]\}$}
  \DisplayProof
    }\\[0.5cm]
    \multicolumn{2}{c}{
      \AxiomC{
  $\+A':~  G\{\Gamma\seqarrowann{n}\Delta,[\Lambda\seqarrowann{m}\Theta,\langle\Sigma\seqarrowann{\ell}\Pi\rangle],\langle\seqarrowann{k}[(\aseq{\Sigma}{\ell}{\Pi})^\circlearrowleft]\rangle\}$
  }
  \LeftLabel{$\star\star$}
  \RightLabel{\rulebc}
  \UnaryInfC{$\stack G\{\Gamma\seqarrowann{n}\Delta,[\Lambda\seqarrowann{m}\Theta,\langle\Sigma\seqarrowann{\ell}\Pi\rangle]\}$}
  \DisplayProof
    }\\[0.7cm]
\multicolumn{2}{c}{$\star$: $k$ and $h$ are not used as annotations in the set-sequent in the conclusion.}\\
\multicolumn{2}{c}{$\star\star$: $k$ is not used as annotations in the set-sequent in the conclusion, 
and }\\
\multicolumn{2}{c}{\qquad 
$\+A '= \+A\cup\{(j,j')~|~\aseq{\Phi}{j}{\Psi}\in^+\aseq{\Sigma}{\ell}{\Pi}~\text{and}~\aseq{\Phi}{j'}{\Psi}\in^+(\aseq{\Sigma}{\ell}{\Pi})^\circlearrowleft\}$.
}
    \end{tabular}
    }
\end{adjustbox}
  \caption{Rules of \ccikann}\label{fig:cik:set:annotated}
\end{figure}

\begin{example}
  Here follows an example of some rule applications in \ccikann:
  \begin{center}
  \small 
    \AxiomC{$(3,6), (4,7); \, \aseq{\square r}{1}{q,\mblock{\aseq{r}{2}{s},\iblock{\aseq{p}{3}{q,\mblock{\aseq{q}{4}{}}}}},\iblock{\aseq{\square r}{5}{\mblock{\aseq{r,p}{6}{q, \mblock{\aseq{q}{7}{\, }}}}}}}$}
    \RightLabel{$(\square_L)$}
    \UnaryInfC{$ (3,6), (4,7);\, \aseq{\square r}{1}{q,\mblock{\aseq{r}{2}{s},\iblock{\aseq{p}{3}{q,\mblock{\aseq{q}{4}{}}}}},\iblock{\aseq{\square r}{5}{\mblock{\aseq{p}{6}{q, \mblock{\aseq{q}{7}{\,}}}}}}}$}
    \RightLabel{\text{(trans)}}
    \UnaryInfC{$ (3,6), (4,7); \, \aseq{\square r}{1}{q,\mblock{\aseq{r}{2}{s},\iblock{\aseq{p}{3}{q,\mblock{\aseq{q}{4}{}}}}},\iblock{\aseq{}{5}{\mblock{\aseq{p}{6}{q, \mblock{\aseq{q}{7}{\,}}}}}}}$}
    \RightLabel{\rulebc}
    \UnaryInfC{$\emptyset ; \,  \aseq{\square r}{1}{q,\mblock{\aseq{r}{2}{s},\iblock{\aseq{p}{3}{q,\mblock{\aseq{q}{4}{\,}}}}}}$}
    \DisplayProof
  \end{center}
\end{example}



\begin{figure}[t]
    \begin{center}
        \begin{tikzpicture}[scale=0.8]
        \tikzstyle{node}=[circle,fill=black,inner sep=1.2pt]
        \node[label= left :{\small{$\aseq{\Gamma}{n}{\Delta}$}}] (1) at (0,0) [node]        {};
        \node[label= right :{\small{$\aseq{\Lambda}{m}{\Theta}$}}] (2) at (2,0) [node]
        {};
        \node[label=  right :{\small{$\aseq{\Sigma}{\ell}{\Pi}$}}] (3) at (2,2) [node]
        {};
        
        \node[label= left :{\small{\textcolor{blue}{$\aseq{\,}{k}{\,}$}}}] (4) at (0, 2) [node]
        {};
        \node[label= above right  :{\small{\textcolor{blue}{$\aseq{\Sigma}{\ell'}{\Pi}$}}}] (5) at (2,2.8) [node]
        {};

        \draw[->] (1) edge[below] node {\footnotesize{$\mblock{\,}$}} (2);
        \draw[->, dashed] (2) edge[right] node{\footnotesize{$\iblock{\,}$}} (3);
        \draw[->, dotted,] (1) edge[left] node{\footnotesize{$\iblock{\,}$}}  (4);
        \draw[->,dotted,] (4) edge[above] node {\footnotesize{$\mblock{\,}$}} (5);
        
        \end{tikzpicture}
    \end{center}
    \caption{Graphical representation of an application of rule \rulebc, as typeset in the text. The arrows marked with $\mblock{\,}$ represent $[\,]$-blocks, and will be interpreted as $R$-relations in the model; the  arrows marked with $\iblock{\,} $ represent $\langle\,\rangle$-blocks, and will be interpreted as $\leq$-relations in the model.  The components $\aseq{\,}{k}{\,}$ and  
    $\aseq{\Sigma}{\ell'}{\Pi}$, as well as the dotted arrows linking them, are introduced in the premiss of the rule. It holds that $(\ell, \ell')\in \+A '$.
    }
    \label{fig:bc}
\end{figure}
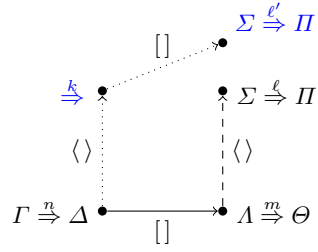

The role of annotations is exemplified in Figure~\ref{fig:bc},  depicting an application of \rulebc. Annotated sequents $\aseq{}{k}{}$ and $\aseq{\Sigma}{\ell'}{\Pi}$ are generated by the rule, and the latter is a duplication of an existing component. According to the algorithm, the original  $\aseq{\Sigma}{\ell}{\Pi}$ needs to be R4-saturated before the rule \rulebc~can be applied to it. Thus, this sequent will not change when proof search continues; however, new formulas might be introduced in the copy  $\aseq{\Sigma}{\ell'}{\Pi}$ introduced by the rule. Annotation are needed to keep track of the dependency between a block and its copies: in the countermodel, we will only keep one copy of each block, while using the other copies to determine the $R$ and $\leq$ relations of the model, ensuring that backward confluence is satisfied. 

We now turn to the countermodel construction. In the following definitions,  
$\+A; S$ is a global-saturated leaf of a derivation  generated by \proc.

\begin{definition}\label{def:annotating-seq}
  Let $T_1,T_2\in^+S$ be annotated by $\ell_1,\ell_2$ respectively. We say $T_2$ \emph{relies on} $T_1$ in 
  $\+A; S$ if $(\ell_1,\ell_2)\in \+A$. A sequent $T\seqin S$ is called \emph{null} in 
  $\+A, S$
  if there is some $T'\seqin S$ such that $T'$ relies on $T$. 
\end{definition}


Moreover, for each $T\seqin S$, it is easy to verify that there is at most one $T'\seqin S$ such that $T'$ relies on $T$.


\begin{definition}
\label{def:model:proof:search}
  The model $\+M_S=(W_S,\leq_S,R_S,V_S)$ determined by 
  $\+A; S$
  is defined as follows:
  \begin{itemize}
    \item $W_S=\{x_{\seq{\Phi}{\Psi}}\mid \aseq{\Phi}{n}{\Psi}\seqin S~\text{is~not~null~in~}S\}$.
    \item For $x_{S_1},x_{S_2}\in W_S$, let 
    $x_{S_1}\leq^0_S x_{S_2}$ iff $S_1\str S_2$ and exactly one of the following holds:
    \begin{itemize}[$-$]
      \item $S_2\in^{\langle\cdot\rangle}_0S_1$;  
      \item there is some null $S_3\seqin S$ with $S_3 \ibin_0 S_1$ and $S_2$ relies on $S_3$;
      \item there are $x_{T_1},x_{T_2}\in W_S$ with $S_1\mbin_0T_1$, $S_2\mbin_0T_2$, and $x_{T_1}\leq^0_S x_{T_2}$; 
      \item $S_1$ is blocked by $S_2$.
    \end{itemize}
  Then take $\leq_S$ to be the reflexive and transitive closure of $\leq_S^0$.
    \item For $x_{S_1},x_{S_2}\in W_S$, $R_S x_{S_1}x_{S_2}$ if $S_2\mbin_0S_1$. 
    \item For each $x_{\Phi\Rightarrow\Psi}\in W_S$, $V_S(x_{\Phi\Rightarrow\Psi})=\{p\mid p\in \Phi\}$.
  \end{itemize}
\end{definition}


In the following proofs, we shall abbreviate $\leq_S$, $\leq^0_S$, $R_S$, $W_S$, $V_S$ as $\leq$, $\leq^0$, $R$, $W$, $V$ respectively 
when $S$ is clear from the context. 
All the sequents in the following proofs are annotated, but for simplicity we omit the annotations.



\begin{lemma}\label{lem:pre-frame-condition}
Let $x_{S_0},x_{S_1},x_{S_2}\in W$, $x_{S_0}Rx_{S_1}$ and $x_{S_1}\leq^0 x_{S_2}$. Then there is $x_{S_3}\in W$ s.t. $x_{S_0}\leq^0 x_{S_3}$ and $x_{S_3}Rx_{S_2}$.
\end{lemma}

\begin{proof}
	Let $S_0 = \seq{\Gamma}{\Delta}$ and $S_1 = \seq{\Sigma}{\Pi}$. 
  From $x_{S_0}Rx_{S_1}$, we have $S_1\mbin_0 S_0$, i.e.,  $\mblock{\seq{\Sigma}{\Pi}}\in\Delta$.
  Since $x_{S_1}\leq^0 x_{S_2}$, by definition $S_1\str S_2$ and exactly one of the following holds:
  \begin{enumerate}[$(i)$]
    \item \label{it:i}$S_2\in^{\langle\cdot\rangle}_0S_1$;
    \item \label{it:ii}There is some null $S_4$ s.t. $S_4 \ibin_0 S_1$ and $S_2$ relies on $S_4$;
    \item \label{it:iii}There are $x_{T_1},x_{T_2}\in W_S$ s.t. $S_1\mbin_0T_1$, $S_2\mbin_0T_2$, $x_{T_1}\leq^0_S x_{T_2}$;
    \item \label{it:iv} $S_1$ is blocked by $S_2$.
  \end{enumerate}
  Let us consider the various cases. 
  If \ref{it:i} holds, we have $\Pi = \Pi',\iblock{S_2}$.  
  Since $S_0$ is saturated with respect to \rulebc, by the saturation condition for \rulebc there is a sequent $\seq{\Phi}{\Psi,[S_2']}\ibin_0 S_0$ s.t. $S_2'$ relies on $S_2$. This makes $S_2$ a null sequent in $S$ and hence 
  $x_{S_2}\notin W$, a contradiction. Thus, this case is ruled out. 
  If \ref{it:ii} holds, then  
  $\Pi =\Pi',\iblock{S_4}$. 
    Since $S_2$ is annotated by $S_4$, 
    it means $S_2$ is the sequent produced by applying \rulebc~to $\mblock{\seq{\Sigma}{\Pi',\iblock{S_4}}}$. 
    By the saturation condition for \rulebc, there is $S_3 =
    \seq{\Phi}{\Psi,\mblock{S_2}}
    \ibin_0 \Delta$. 
    It holds $x_{S_0}\leq^0 x_{S_3}$ and $x_{S_3}Rx_{S_2}$. 
    If \ref{it:iii} holds then, since by definition the $R$-predecessor of each world is unique, $S_0$ is exactly $T_1$ and $x_{T_2}$ is $x_{S_3}$ as required. 
    In case \ref{it:iv} holds, by definition of blocked sequent, we have $S_1\ibin S_2$. However, from $x_{S_0}Rx_{S_1}$ we have $S_1\mbin S_0$. This is a  contradiction, and this case is also ruled out. 
\end{proof}

\begin{proposition}\label{bc-fc-property}
  $\+M_S$ satisfies both $(\bc)$ and $(\fc)$ conditions.
\end{proposition}

\begin{proof}
  Case ($\fc$) immediately follows from the definition of $\str$. 
  To show that ($\bc$) is satisfied, take arbitrary $x_{S},x_{T},x_{T'}\in W$ with $x_{S}Rx_{T}$ and $x_{T}\leq x_{T'}$. We need to show that there is a $x_{S'}$ with  $x_{S}\leq x_{S'}$ and $x_{S'}Rx_{T'}$. We reason by induction on the number $n$ of worlds in the $\leq$-chain between $x_{T}$ and $x_{T'}$. If $n= 0$,  $x_{T}=x_{T'} $, and we take $x_{S'} = x_{T}$. The case $n=1$ is treated in Lemma~\ref{lem:pre-frame-condition}. Suppose $n > 1$. Then,   there is some $x_{T_1}\in W$ s.t.  $x_{T_1}\neq x_{T'}\neq x_{T}$, $x_{T}\leq x_{T_1}$ and $x_{T_1}\leq^0 x_{T'}$.  By IH, there is $x_{S_1}\in W$ such that $x_{S}\leq x_{S_1}$ and $Rx_{S_1}x_{T_1}$.  By Lemma~\ref{coro:no-inf-chain} there is $x_{S_2}\in W$ s.t. $x_{S_1}\leq x_{S_2}$ and $x_{S_2}Rx_{T'}$. By transitivity of $\leq$,  $x_{S}\leq x_{S_2}$. We take $x_{S'}= x_{S_2}$ and we are done. 
%
\end{proof}

\begin{lemma}[Truth Lemma]\label{lem:truth-lemma}
  Let $S$ be a global-saturated enriched sequent and $\+M_{S}$ defined as above. 
  The following hold: 
  \begin{enumerate}[$(a)$]
      \item \label{it:truth_lemma:1} 
      If $A\in\Phi$, then $\+M_{S},x_{\Phi\Rightarrow\Psi}\Vdash A$;
      \item \label{it:truth_lemma:2} 
      If $A\in\Psi$, then $\+M_{S},x_{\Phi\Rightarrow\Psi}\nVdash A$.
  \end{enumerate}
\end{lemma}

\begin{proof}
  By induction on the complexity of $A$. We only show the case of $A =B\supset C$ (some other cases can be found in~\cite{csl-2024-fik}). 
  We abbreviate $x_{\Phi\Rightarrow\Psi}$ as $x$. 

    In case \ref{it:truth_lemma:1}, $B\supset C\in\Phi$. Then, $\Phi\Rightarrow\Psi$ satisfies the saturation condition associated with $(\supset_L)$ for $A$ regardless of whether the sequent is blocked or not. 
    Assume for the sake of contradiction that $x\nVdash B\supset C$. Then there exists a world $x_{\Omega\Rightarrow\Xi}$ in $W_S$ with $x\leq x_{\Omega\Rightarrow\Xi}$ s.t. $x_{\Omega\Rightarrow\Xi}\Vdash B$ and $x_{\Omega\Rightarrow\Xi}\nVdash C$. 
    Since $x\leq x_{\Omega\Rightarrow\Xi}$, by definition, we have $\Phi\Rightarrow\Psi\subseteq^S \Omega\Rightarrow\Xi$, which implies $B\supset C\in\Omega$. 
    Moreover, since $\Omega\Rightarrow\Xi$ satisfies the saturation condition associated with $(\supset_L)$ for $B\supset C$, we 
    have $B\in\Xi$ or $C\in\Omega$.
    By IH, either $x_{\Omega\Rightarrow\Xi}\nVdash B$ or $x_{\Omega\Rightarrow\Xi}\Vdash C$. Both cases lead to a contradiction. 
  
    In case \ref{it:truth_lemma:2}, $B\supset C\in\Psi$. 
    Aiming at a contradiction, suppose that $x\Vdash B\supset C$. Then for every $x'$ with $x\leq x'$, $x'\Vdash B$ implies $x'\Vdash C$. 
    We distinguish cases according to whether $\Phi\Rightarrow\Psi$ is blocked and also whether a sequent involved in the $(\supset_R)$ saturation is null or not. 
    %
    
    If $\Phi\Rightarrow\Psi$ is not blocked, then the sequent satisfies the saturation condition associated with $(\supset_R)$ for $B\supset C$, and we have two possible cases.  
    $(i)$ $B\in \Phi$ and $C\in\Psi$. By IH, it follows $x\Vdash B$ and $x\nVdash C$. By reflexivity of $\leq$, we have $x\leq x$, then $x\Vdash B$ entails that $x\Vdash C$ as well, a contradiction.
     $(ii)$ 
      There is a block $\langle\Lambda\Rightarrow\Theta\rangle\in\Psi$ such that $B\in\Lambda$ and $C\in\Theta$. Since $\seq{\Lambda}{\Theta}$ it is saturated with (trans) and (inter) and by Proposition \ref{prop:sequent-inclusion}, we have $\Phi\Rightarrow\Psi\subseteq^S\Lambda\Rightarrow\Theta$. 
      In this case we need to distinguish whether $\Lambda\Rightarrow\Theta$ is null or not. 
      If $\Lambda\Rightarrow\Theta$ is not null, by definition, we have $x\leq x_{\Lambda\Rightarrow\Theta}$. Since $B\in\Lambda$, by IH we have $x_{\Lambda\Rightarrow\Theta}\Vdash B$, and hence $x_{\Lambda\Rightarrow\Theta}\Vdash C$. Moreover, as $C\in\Theta$, by IH it follows that $x_{\Lambda\Rightarrow\Theta}\nVdash C$, a contradiction. 
      Otherwise, if $\seq{\Lambda}{\Theta}$ is null we have $x_{\Lambda\Rightarrow\Theta}\notin W$ and there is $x_{\Lambda_1\Rightarrow\Theta_1}\in W$ s.t. $\Lambda_1\Rightarrow\Theta_1$ relies on $\Lambda\Rightarrow\Theta$. By 
      definition, 
      $x\leq x_{\seq{\Lambda_1}{\Theta_1}}$.  
      We obtain a contradiction by applying the same argument to $\Lambda_1\Rightarrow\Theta_1$. 
      
    Otherwise, if $\Phi\Rightarrow\Psi$ is blocked, then the sequent does not satisfy the saturation condition associated with $(\supset_R)$ for $B\supset C$. 
    By definition, there is an unblocked sequent $\Sigma\Rightarrow\Pi\in^+S$ s.t. $\Phi\Rightarrow\Psi$ is blocked by it. 
    Then we have $\Sigma\Rightarrow\Pi\simeq \Phi\Rightarrow\Psi$, which implies $\Pi^\sharp=\Psi^\sharp$, so $B\supset C\in\Pi$. 
    By definition 
    we have both $x\leq x_{\Sigma\Rightarrow\Pi}$ and $x_{\Sigma\Rightarrow\Pi}\leq x$. 
    Since $\Sigma\Rightarrow\Pi$ is R3-saturated, it already satisfies the saturation condition associated with $(\supset_R)$ for $B\supset C$. 
    By applying the same argument in the previous case  to $\Sigma\Rightarrow\Pi$ we have that either $(i)$ $x_{\Sigma\Rightarrow\Pi}\Vdash B$ and $x_{\Sigma\Rightarrow\Pi}\nVdash C$, in which case we conclude similarly as before, or $(ii)$ there is a block $\langle\Lambda\Rightarrow\Theta\rangle\in\Pi$ s.t. $B\in\Lambda$ and $B\in\Theta$. 
      If $\Lambda\Rightarrow\Theta$ is null, it means the block $[\seq{\Sigma}{\Pi',\iblock{\seq{\Lambda}{\Theta}}}]$ activates an \rulebc~application. 
      Since $\seq{\Phi}{\Psi}$ is blocked by $\seq{\Sigma}{\Pi}$, we see 
      $\seq{\Phi}{\Psi}\ibin\seq{\Sigma}{\Pi}$ which must also activate an \rulebc~application and hence make $\seq{\Phi}{\Psi}$ a null sequent. Since $x\in W$, we see this leads to a contradiction. 
      Thus $\Lambda\Rightarrow\Theta$ is not null and we have $x_{\Lambda\Rightarrow\Theta}\Vdash B$ and $x_{\Lambda\Rightarrow\Theta}\nVdash C$ as well as $x_{\Sigma\Rightarrow\Pi}\leq x_{\Lambda\Rightarrow\Theta}$.
     Since  $x\leq x_{\Sigma\Rightarrow\Pi}$ and $x_{\Sigma\Rightarrow\Pi}\leq x$, we have $x\leq x_{\Lambda\Rightarrow\Theta}$.
     Since $x_{\seq{\Lambda}{\Theta}}\Vdash B$, which implies $x_{\seq{\Lambda}{\Theta}}\Vdash C$ as well. Recall $x_{\Lambda\Rightarrow\Theta}\nVdash C$, a contradiction.
\end{proof}

As a result, we obtain the completeness of \textbf{C}\cik.

\begin{theorem}[Completeness]
  If $A$ is valid in {\rm \ik},  it is provable in {\rm \textbf{C}\cik}.
\end{theorem}

\begin{example}\label{ex:countermodel}
    We run $\mathsf{ProofSearch}(\neg\Diamond\neg p\supset\square p)$, obtaining the  global-saturated enriched sequent $S$ = $(4,6); \, \seqarrowann{0} \neg\Diamond\neg p\supset\square p, \langle T \rangle$, where $T$ is:\\[-0.2cm] 
   
    \begin{adjustbox}{max width= \textwidth}
    $
    \neg\Diamond\neg p  \overset{1}{\Rightarrow}  \Diamond\neg p,\square p, 
    \iblock{\aseq{\neg\Diamond\neg p}{2}{\Diamond\neg p,\mblock{\aseq{}{3}{p,\neg p,\iblock{\aseq{p}{4}{\bot}}}},\iblock{\aseq{\neg\Diamond\neg p}{5}{\Diamond\neg p,\mblock{\aseq{p}{6}{\neg p,\bot}}}}}}.
    $
    \end{adjustbox}\\
    \normalsize
    
    We denote by $S_i$ the sequent annotated by $i$: e.g., $S_4$ is $\aseq{p}{4}{\bot}$ and $S_6$ is $\aseq{p}{6}{\neg p,\bot}$. It holds that $S_6$ relies on $S_4$, thus $S_4$ is a null sequent. The model $\+M_S$ is defined as follows: 
    $W=\{x_{S_0},x_{S_1},x_{S_2},x_{S_3},x_{S_5},x_{S_6}\}$; $x_{S_0}\leq x_{S_1}\leq x_{S_2}\leq x_{S_5},x_{S_3}\leq x_{S_6}$ (plus reflexive closure); $Rx_{S_2}x_{S_3}, Rx_{S_5}x_{S_6}$; and $V(x_{S_0})=V(x_{S_1})=V(x_{S_2})=V(x_{S_3})=V(x_{S_5})=\emptyset$ and $V(x_{S_6})=\{p\}$. 
    It is easy to verify that $x_{S_0}\Vdash\neg\Diamond\neg p$ and $x_{S_0}\nVdash\square p$, whence $x_{S_0}\nVdash\neg\Diamond\neg p\supset\square p$.  
    An illustration of the countermodel can be found in Figure \ref{ex:countermodel-pic}.
\end{example}

\begin{figure}[!t]
      
      
      
      \begin{center}
      \begin{tikzpicture}[scale=0.8]
      \tikzstyle{node}=[circle,fill=black,inner sep=1.2pt]
      \tikzstyle{nodegray}=[circle,fill=lightgray,inner sep=1.2pt]

    \node[node,label= left :{$x_{S_0}$}] (0) at (0,-4) []        {};
      
      \node[node, label= left :{$x_{S_1}$}] (1) at (0,-2) []        {};
      
      \node[node,label= left :{$x_{S_2}$}] (2) at (0,0) []        {};
      \node[node,label= below :{$x_{S_3}$}] (3) at (2,0) []
      {};
      \node[nodegray,label= left:{\textcolor{lightgray}{$x_{S_4}$}}] (4) at (2,2) [] {};
      
      \node[node,label= left :{\textcolor{black}{$x_{S_5}$}}] (5) at (0, 2) []
      {};
      \node[node,label= above  :{\textcolor{black}{$x_{S_6}\Vdash p$}}] (6) at (2,2.8) []
      {};
      
      \draw[->] (2) edge node [below] {\textcolor{black}{$R$}} (3);
      \draw[->, dashed,lightgray] (3) -- (4);
      \draw[->, dashed] (2) -- (5);
      \draw[->, dashed] (1) -- (2);
      \draw[->,dashed] (0) edge node [left] {\textcolor{black}{$\leq$}} (1);
      \draw[->] (5) -- (6);
      \draw[->,dashed] (3) to [in=20, out=-20](6);
      \end{tikzpicture}
  
  \end{center}
  \caption{
  A graphical representation of Example \ref{ex:countermodel}. The dashed arrows represent $\leq$ relations, and solid arrows represent the $R$ relations. World $x_{S_4}$, corresponding to the sequent $S_4$, is \emph{not} part of the model, because $S_4$ is a null sequent. }
  \label{ex:countermodel-pic}
\end{figure}
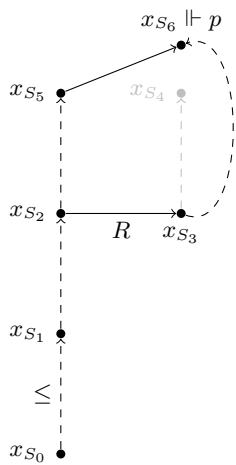

\section{Relations with other proof systems
for \ik}\label{sec:simulation}

In this section we investigate the relation between \cik \ and other existing calculi for \ik, namely the fully labelled calculus presented in \cite{Marin:et:al:2021} and the  nested calculus in \cite{2013cut}. 
These results provide alternative completeness proofs for \cik. 
We shall also discuss the relations between our calculus and the 
labelled calculus introduced in~\cite{simpson1994proof}, the first calculus known for \ik, 
and the
calculus from~\cite{GorePT10}. 

\subsection{From fully labelled to bi-nested sequents}

We first consider the relation between \cik and the fully labelled sequent calculus \labik presented in \cite{Marin:et:al:2021}. This is not the first labelled calculus for \ik, a labelled calculus was already proposed in \cite{simpson1994proof} and we will discuss it later. 


As 
usual when defining labelled sequents, the first step is to enrich the language of the logic by adding semantic information. Given a countably enumerable set of variables (the \emph{labels}) $x, y, z, \dots$, intuitively representing worlds of bi-relational models, the \emph{labelled formulas} of \labik are taken to be \emph{relational atoms}, of the form  $xRy$ or  $x\leq y$, and objects of the form $x:A$, where $A$ is a formula of the non-labelled language. The two kinds of relational atoms $xRy$ and $x\leq y$ characterize the pre-order and accessibility relation in a bi-relational frames, respectively. \footnote{The proof system introduced in~\cite{Marin:et:al:2021} is called ``fully labelled'' to distinguish it form Simpson's labelled calculus, where only one relational atom, $xRy$, appeared. We briefly discuss Simpson's calculus in Section~\ref{sec:simpson_gore}.}

A \emph{labelled sequent} is a triple  $\+R,\seq{\Gamma}{\Delta}$, for $\+R, \Gamma$ and $\Delta$ multisets of labelled formulas, where relational atoms are only allowed to occur in $\+R$, and $\Gamma$ and $\Delta$ only contain labelled formulas of the form $x:A$. 
Under this setting, rules in the deductive system \labik directly mimic the semantics, for example,

\begin{prooftree}
    \AXC{$\+R,x\leq y,yRz,\seq{\Gamma}{\Delta,z:A}$}
    \RightLabel{$y,z$ fresh}
    \LeftLabel{$\square_\mathsf{R}$}
    \UIC{$\+R,\seq{\Gamma}{\Delta,x:\square A}$}
\end{prooftree}
where ``$y, z$ fresh'' means that the two labels do not occur in the conclusion of the rule. 
Frame conditions are also encoded directly, 

\begin{center}
    \begin{adjustbox}{max width=\textwidth}
        \begin{tabular}{ccc}
            \AXC{$\+R, x R y, y \leq z, x \leq u, u R z, \Gamma \Longrightarrow \Delta$}
            \LeftLabel{$\mathsf{F1}$}
            \RightLabel{$u$ fresh}
            \UIC{$\+R, x R y, y \leq z, \Gamma \Longrightarrow \Delta$}
            \DP
            &
            \quad
            &
            \AXC{$\+R, x R y, x \leq z, y \leq u, z R u, \Gamma \Longrightarrow \Delta$}
            \LeftLabel{$\mathsf{F2}$}
            \RightLabel{$u$ fresh}
            \UIC{$\+R, x R y, x \leq z, \Gamma \Longrightarrow \Delta$}
            \DP
        \end{tabular}
    \end{adjustbox}
\end{center}
where $\mathsf{F1}$ and $\mathsf{F2}$ are what we call $\bc$ and $\fc$ respectively.


To establish a correspondence between blocks and labels, it is convenient to consider \ccikann, the annotated version of \cik presented in Figure~\ref{fig:cik:set:annotated}. As observed in Section~\ref{sec:completeness}, \cik and \ccikann are equivalent, and thus our translation results can be extended to \cik as well. 


Now we show how to simulate a derivation in the fully labelled calculus into another one in \ccikann. As it is usual when translating a labelled calculus into a non-labelled one, we have that not all labelled sequents can be translated into annotated (and, ultimately, bi-nested) ones. 
This is because the structure of the labels occurring in a labelled sequent can be an arbitrary graph, which may not have a counterpart in bi-nested sequent. 
Therefore, our translations will cover only a subset of derivations in \labik, as specified later. 
Besides, we can regard each multiset $\+R$ as giving rise to a graph in the obvious way. So we can talk about $(\leq,R)$-paths and $\leq$-maximal paths. 
The full definition of 
graph underlying a sequent 
can be found in \cite{wollic-2024-ik}.

Let $\mathfrak{S}=\+R,\seq{\Gamma}{\Delta}$ be a two-sided labelled sequent. The set of labels occurring in $\+R$ is denoted by $\lab(\+R)$. Let $\Omega=\{A_1,\ldots,A_k\}$ be a finite set of formulas and $x$ a label, we denote the set of labelled formulas $x:A_1, \ldots, x:A_k$ collectively by $x:\Omega$. 

\begin{definition}
    Let $\mathfrak{S}=\+R,\seq{\Phi}{\Psi}$ be a two-sided labelled sequent. 
    For $x,y\in \lab(\+R)$, we say $y$ is \textit{connected} to $x$, if there is a $(\leq,R)$-path from $x$ to $y$ in $\+R$.
\end{definition}

\begin{definition}[Tree-like labelled sequent] 
    Let $\mathfrak{S}=\+R,\seq{\Phi}{\Psi}$ be a two-sided labelled sequent. $\mathfrak{S}$ is called \emph{tree-like} if:
    \begin{enumerate}
        \item there is a unique label $x$ occurring in $\seq{\Phi}{\Psi}$
        s.t. any other label $y\in\lab(\+R)$ is connected to it; and 
        \item labels occurring in $\seq{\Phi}{\Psi}$ all belong to $\lab(\+R)\cup\{x\}$.
    \end{enumerate}
    Label $x$ satisfying condition (1) is called the \textit{rooted label} of $\mathfrak{S}$.
\end{definition}

For instance, both $\seq{}{x:A}$ and $x\leq u,uRz,\seq{z:A}{x:A\wedge B}$ are tree-like while $xRy,z\leq y,\seq{x:A}{z:B}$ and $xRy,\seq{x:A}{z:B}$ are not. 

\begin{proposition}
    If a labelled sequent $\mathfrak{S}$ occurs in a  \labik derivation 
    having sequent $\seq{}{x:A}$ as conclusion, then it must be tree-like and have $x$ as its rooted label. 
\end{proposition}

With definitions above, we can write a tree-like labelled sequent $\mathfrak{S}=\+R,\seq{\Phi}{\Psi}$ as 
$$
\+R,\bigcup\limits_{i=0}^{n}x_i:\Gamma_i\Rightarrow\bigcup\limits_{i=0}^{n}x_i:\Delta_i
$$
for some $n\in\mathbb{N}$, where 
$\Gamma_i$'s and $\Delta_i$'s are multi-set of label-free formulas, 
$\lab(\+R)\cup\{x_0\}=\{x_i\}_{i=0}^n$ and $x_0$ is the rooted label of $\mathfrak{S}$. 

We demonstrate how to translate $\mathfrak{S}$ into an annotated bi-nested sequent $S$, 
which is made by 
arranging each $\seq{\Gamma_i}{\Delta_i}$ into nested structures according to the relational atoms occurring in $\+R$. 
For each $\seq{\Gamma_i}{\Delta_i}$, the translation $\tr{\seq{\Gamma_i}{\Delta_i}}$ is defined inductively on the converse of $R$ and $\leq$ according to the relational atoms in $\+R$,

\begin{itemize}
    \item if $x_i$ is both $\leq$-maximal and $R$-maximal, let $\tr{\seq{\Gamma_i}{\Delta_i}}=\aseq{\Gamma_i}{i}{\Delta_i}$;
    \item if $x_i$ has direct $\leq$-successors $x_{j_1},\ldots,x_{j_m}$ 
    and direct $R$-successors $x_{k_1},\ldots,x_{k_n}$, let $\tr{\seq{\Gamma_i}{\Delta_i}}$ be 
    
    \begin{adjustbox}{max width=.9\textwidth}
        $\aseq{\Gamma_i}{i}{\Delta_i,\iblock{\tr{\seq{\Gamma_{j_1}}{\Delta_{j_1}}}},\ldots,\iblock{\tr{\seq{\Gamma_{j_m}}{\Delta_{j_m}}}},\mblock{\tr{\seq{\Gamma_{k_1}}{\Delta_{k_1}}}},\ldots,\mblock{\tr{\seq{\Gamma_{k_n}}{\Delta_{k_n}}}}}$
    \end{adjustbox}
\end{itemize}

Lastly, since $x_0$ is the unique minimal element of $R$ and $\leq$, we just take $S$ to be $\tr{\seq{\Gamma_0}{\Delta_0}}$.  

\begin{example}
    For $\mathfrak{S}=x_0\leq x_1,x_1Rx_2,\seq{x_2:A}{x_0:A\wedge B}$, we have $\tr{\mathfrak{S}}= \ \aseq{}{0}{A\wedge B,\iblock{\aseq{}{1}{\mblock{\aseq{}{2}{A}}}}}$.
\end{example}


Now, together with the translation we defined above, we claim

\begin{theorem}\label{thm:label-to-bi-nested}
    If a tree-like sequent $\mathfrak{S}$ is derivable in \labik, then its translation $\tr{\mathfrak{S}}$ is derivable in \ccikann.
\end{theorem}

\begin{proof}
    The proof is done by induction on the height of a derivation. We verify each rule application can be simulated by a derivation in \ccikann. 

    Recall the following rules are admissible in \cik and its variants,

\begin{center}
    \begin{adjustbox}{max width=\textwidth}
      \begin{tabular}{ccccc}
        \AxiomC{$G\{\Gamma\Rightarrow\Delta\}$}
        \RightLabel{$(w_{R})$}
        \UnaryInfC{$G\{\Gamma\Rightarrow\Delta,\+O\}$}
        \DisplayProof
        &
        \quad
        &
        \AxiomC{$G\{\Gamma\Rightarrow\Delta\}$}
        \RightLabel{$(w_{L})$}
        \UnaryInfC{$G\{A,\Gamma\Rightarrow\Delta\}$}
        \DisplayProof
        &
        \quad
        &
        \AxiomC{$G\{\Gamma\Rightarrow\Delta,\iblock{\seq{}{}}\}$}
        \RightLabel{$(\iblock{\emptyset})$}
        \UnaryInfC{$G\{\Gamma\Rightarrow\Delta\}$}
        \DisplayProof
      \end{tabular}
    \end{adjustbox}
\end{center}
    where in $(w_{R})$, $\+O$ can be either a formula or a block. 

    For axioms and logical rules, we consider $(\mathsf{id}), (\supset_{\mathsf{L}}),(\square_{\mathsf{R}}),(\Diamond_{\mathsf{R}})$ as examples. Other rules are trivial. 
    We omit occurrences of $\+R, \Gamma,\Delta$
    for labelled rules on the left and annotation set (which is always empty) for bi-nested sequents on the right to simplify the proof.

\begin{center}
\begin{adjustbox}{max width=\textwidth}
\small{
    \begin{tabular}{ccc}
        \AxiomC{}
        \LeftLabel{$(\mathsf{id})$}
        \UnaryInfC{$x\leq y,\seq{x:p}{y:p}$}
        \DisplayProof
        &
        \qquad $\leadsto$ \qquad
        &
        \AxiomC{}
        \RightLabel{(id)}
        \UnaryInfC{$\aseq{}{x}{p,\iblock{\aseq{p}{y}{p}}}$}
        \RightLabel{(trans)}
        \UnaryInfC{$\aseq{p}{x}{\iblock{\aseq{}{y}{p}}}$}
        \DisplayProof
        \\
        &&\\
        \multicolumn{3}{l}{
        \AxiomC{$x\leq y,\seq{x:A\supset B}{y:A}$}
        \AxiomC{$x\leq y,\seq{y:B}{}$}
        \LeftLabel{$(\supset_{\mathsf{L}})$}
        \BinaryInfC{$x\leq y,\seq{x:A\supset B}{}$}
        \DisplayProof
        \qquad $\leadsto$ \qquad}
        \\
         &&\\
        \multicolumn{3}{r}{
        \qquad $\leadsto$ \qquad
        \AxiomC{$\aseq{A\supset B}{x}{\iblock{\aseq{}{y}{A}}}$}
        \RightLabel{($w_L$)}
        \UnaryInfC{$\aseq{A\supset B}{x}{\iblock{\aseq{A\supset B}{y}{A}}}$}
        \AxiomC{$\aseq{}{x}{\iblock{\aseq{B}{y}{}}}$}
        \RightLabel{($w_L$)}
        \UnaryInfC{$\aseq{A\supset B}{x}{\iblock{\aseq{B}{y}{}}}$}
        \RightLabel{($\supset_L$)}
        \BinaryInfC{$\aseq{A\supset B}{x}{\iblock{\aseq{A\supset B}{y}{}}}$}
        \RightLabel{(trans)}
        \UnaryInfC{$\aseq{A\supset B}{x}{\iblock{\aseq{}{y}{}}}$}
        \DisplayProof
        }
        \\
        &&\\
        \AxiomC{$x\leq u,uRz\seq{}{z:A}$}
        \LeftLabel{$(\square_\mathsf{R})$}
        \RightLabel{$u,z$ fresh}
        \UnaryInfC{$\seq{}{x:\square A}$}
        \DisplayProof
        &
        \qquad $\leadsto$ \qquad
        &
        \AxiomC{$\aseq{}{x}{\iblock{\aseq{}{u}{\mblock{\aseq{}{z}{A}}}}}$}
        \RightLabel{$(\square_R)$}
        \UnaryInfC{$\aseq{}{x}{\square A}$}
        \DisplayProof\\
        &&\\
        \AxiomC{$xRy\seq{}{x:\Diamond A,y:A}$}
        \LeftLabel{$(\Diamond_\mathsf{R})$}
        \UnaryInfC{$xRy\seq{}{x:\Diamond A}$}
        \DisplayProof
        &
        \qquad $\leadsto$ \qquad
        &
        \AxiomC{$\aseq{}{x}{\Diamond A,\mblock{\aseq{}{y}{A}}}$}
        \RightLabel{$(\Diamond_R)$}
        \UnaryInfC{$\aseq{}{x}{\Diamond A,\mblock{\aseq{}{y}{}}}$}
        \DisplayProof
    \end{tabular}
    }
\end{adjustbox}
\end{center}

Labelled rules capturing the frame conditions of forward and backward confluence can be simulated as follows, where $u$ is a fresh label:

\begin{center}
    \begin{adjustbox}{max width=\textwidth}
    \small{
        \begin{tabular}{c}
            \AxiomC{$xRy,y\leq z,x\leq u,uRz,,\seq{y:A,z:C}{y:B,z:D}$}
            \LeftLabel{$(\mathsf{F1})$}
            \UnaryInfC{$xRy,y\leq z,,\seq{y:A,z:C}{y:B,z:D}$}
            \DisplayProof
            \quad
            $\leadsto$ 
            \\
            \\
            $\leadsto$ 
            \quad
            \AxiomC{$\aseq{}{x}{\mblock{\aseq{A}{y}{B,\iblock{\aseq{C}{z}{D}}}},\iblock{\aseq{}{u}{\mblock{\aseq{C}{z}{D}}}}}$}
            \RightLabel{\rulebc}
            \UnaryInfC{$\aseq{}{x}{\mblock{\aseq{A}{y}{B,\iblock{\aseq{C}{z}{D}}}},}$}
            \DisplayProof
            \\
            \\
            \\
            \AxiomC{$xRy,x\leq z,y\leq u,zRu,\seq{y:A,z:C}{y:B,z:D}$}
            \LeftLabel{$(\mathsf{F2})$}
            \UnaryInfC{$xRy,x\leq z,\seq{y:A,z:C}{y:B,z:D}$}
            \DisplayProof
            \quad
            $\leadsto$ \\
            \\
            $\leadsto$
            \qquad
            \AxiomC{$\aseq{}{x}{\mblock{\aseq{A}{y}{B,\iblock{\aseq{}{u}{}}}},\iblock{\aseq{C}{z}{D,\mblock{\aseq{}{u}{}}}}}$}
            \RightLabel{$(\iblock{\emptyset})$}
            \UnaryInfC{$\aseq{}{x}{\mblock{\aseq{A}{y}{B}},\iblock{\aseq{C}{z}{D,\mblock{\aseq{}{u}{}}}}}$}
            \RightLabel{$(w_L)$}
            \UnaryInfC{$\aseq{}{x}{\mblock{\aseq{A}{y}{B}},\iblock{\aseq{C}{z}{D,\mblock{\aseq{A}{u}{}}}}}$}
            \RightLabel{\rulefc}
            \UnaryInfC{$\aseq{}{x}{\mblock{\aseq{A}{y}{B}},\iblock{\aseq{C}{z}{D}}}$}
            \DisplayProof
            \\
        \end{tabular}
        }
    \end{adjustbox}
\end{center}

Note that in a more general application of $(\mathsf{F2})$, there may be other labels related to $y$, which will result in nested modal blocks. In this case, the nested structure can still be copied into the empty modal block annotated by $u$ in order to fit the form of \rulefc, this is ensured by $(w_L)$. Similarly, in $(\mathsf{F1})$ there may be other labels related to $z$, the resulting nested structure will be copied in the implication block annotated by $u$. This is due to the inductive feature of translation. 

Lastly, for $(\leq_\mathsf{rf})$ and $(\leq_\mathsf{tr})$, according to the translation, the premise and conclusion become the same in bi-nested form, so they are simulated by the same sequent.  
\end{proof}

Recall a derivation in \ccikann can be easily simulated by one in \cik, it is done by first erasing all the occurrences of annotations and then ignoring the principal formulas in each rule application. Consequently, we obtain the following:

\begin{corollary}
    If a tree-like sequent $\mathfrak{S}$ is derivable in \labik, then the sequent obtained by erasing all the annotations from its translation $\tr{\mathfrak{S}}$ is derivable in \cik.
\end{corollary}


Conversely, it is easy to see that the calculus \cik can be simulated by \labik. The translation from bi-nested sequents to labelled sequents is immediate: formulas occurring in different blocks are labelled with different labels, and the $\iblock{\cdot}$- and $\mblock{\cdot}$-structure give rise to $\leq$- and $R$-relational atoms. Then, every bi-nested rule can be simulated by (possibly many) labelled rules. 

The fully labelled calculus contains all the semantic ingredients necessary to interpret the bi-nested sequents of \cik. And indeed, to every bi-nested sequent there corresponds a labelled sequent. However, the converse is not true:  
according to Theorem \ref{thm:label-to-bi-nested}, we can translate tree-like labelled sequents in the bi-nested formalism. And indeed, while labelled sequents can encode any kind of graph structure built from the relations $R$ and $\leq$, bi-nested sequents can only encode tree structures over $\iblock{\cdot}$- and $\mblock{\cdot}$-blocks (and thus, over $R$ and $\leq$ in the semantics).



\subsection{From nested sequents to bi-nested sequents}

We now turn our attention to the multi-conclusion nested sequent calculus for $\ik$ from \cite{kuznets2019maehara}, called $\nikm$. We shall show how to simulate $\nikm$ proofs into the bi-nested calculus \cik. This provides an alternative proof of completeness for \cik. 
Moreover, in \cite{kuznets2019maehara} the authors show how $\nikm$ simulates the proofs in the single-conclusion nested sequent calculus for $\ik$, from  \cite{2013cut,marin2014label}. Thus, \cik indirectly simulates proofs in the single-conclusion nested calculus. 

In a nutshell, the proof system $\nikm$ from \cite{kuznets2019maehara} enriches a multi-conclusion sequent calculus for intuitionistic logic with the structural connective $\mblock{\cdot}$, representing the $R$-relation of bi-relational models. Multi-conclusion nested sequents are trees of multi-conclusion sequents generated by $\mblock{\cdot}$. 
Instead of employing the standard two-sided notation for sequents, the nested sequents of $\nikm$ differentiate antecedent and consequent by assigning \emph{polarities} to formulas. That is, to each formula occurring in a nested sequent is associated an extra bit of information: either $\bullet$ for the \emph{input/left polarity}, or $\circ$ for the \emph{output/right polarity}. \footnote{The use of polarities dates back to \cite{Lamarche:2003}.} Then, a \emph{polarised multi-conclusion nested sequent} (\emph{polarised nested sequent} for short) is either the empty multiset or the following object, where $A_1, \dots, A_k$, $B_1, \dots, B_h$ are formulas and $\Sigma_1, \dots,\Sigma_n$ are polarised nested sequents:
$$
\lef{A_1},\dots,\lef{A_k},\rig{B_1}, \dots, \rig{B_h}, \mblock{\Sigma_1}, \dots, \mblock{\Sigma_n}
$$
Given a polarised nested sequent $\Sigma$, we denote by $\fm{\Sigma}$ the ``formula part'' of $\Sigma$, namely the multiset of formulas $\lef{A_1},\dots,\lef{A_k},\rig{B_1}, \dots, \rig{B_h}$ occurring in $\Sigma$.  

Similarly to bi-nested sequents, polarised nested sequents do not have a formula interpretation in the language of $\ik$. 

It is easy to see that to each polarised nested sequent  there corresponds a \emph{flat} bi-nested sequent, i.e., a bi-nested sequent in which no $\iblock{\cdot}$-block occurs. 
Conversely, our bi-nested sequents can be written employing polarised formulas. 

The notion of context $\cseqnes{}$ for polarised nested sequents is slightly different from the notion of context for bi-nested sequents. For bi-nested sequents, contexts are bi-nested sequents with a hole, where the hole can only occur in specific parts of the sequent (in the consequent, and within a block). 
Contexts of bi-nested sequents can only be filled with other bi-nested sequents. 
Instead, a \emph{polarised context} is a polarised nested sequent with a hole, where the hole takes the place of \emph{any} formula in the sequent. Moreover, a hole in a polarised context can be filled by any polarised nested sequent. 
Let $\cseqpdown{}$ denote the polarised context $\cseqnes{}$ with all output formulas removed. 
The rules of $\nikm$ from~\cite{kuznets2019maehara}  are showcased in Figure~\ref{fig:nikm}.

\begin{figure}[t]
  \begin{center}
    \begin{tabular}{c c}
      \AxiomC{}
      \RightLabel{\rm ($\lef\bot$)}
      \UnaryInfC{$ \cseqnes{\lef\bot}$}
      \DisplayProof
      &
         \AxiomC{}
      \RightLabel{\rm ($\text{id}$)}
      \UnaryInfC{$ \cseqnes{\lef p, \rig p}$}
      \DisplayProof
    \\[0.3cm]    
     \AxiomC{$ \cseqnes{\lef A, \lef B}$}
    \RightLabel{\rm ($\lef\wedge$)}
    \UnaryInfC{$ \cseqnes{\lef{A \land B}}$}
    \DisplayProof
    & 
      \AxiomC{$ \cseqnes{\rig A} $}
    \AxiomC{$ \cseqnes{\rig B} $}
    \RightLabel{\rm ($\rig\wedge$)}
    \BinaryInfC{$ \cseqnes{\rig{A\land B} }$}
    \DisplayProof\\[0.5cm]
     \AxiomC{$ \cseqnes{\lef A} $}
    \AxiomC{$ \cseqnes{\lef B} $}
    \RightLabel{\rm ($\lef\vee$)}
    \BinaryInfC{$\cseqnes{\lef{A \lor B}} $}
    \DisplayProof
    & 
     \AxiomC{$ \cseqnes{\rig A, \rig B} $}
    \RightLabel{\rm ($\rig \vee$)}
    \UnaryInfC{$ \cseqnes{\rig{A \lor B}} $}
    \DisplayProof\\[0.5cm]
    \AxiomC{$ \cseqnes{ \rig{A}} $}
  \AxiomC{$ \cseqnes{\lef B} $}
  \RightLabel{\rm ($\lef\supset$)}
  \BinaryInfC{$ \cseqnes{\lef{A \supset B}} $}
  \DisplayProof
   & 
    \AxiomC{$ \cseqpdown{\lef A, \rig B} $}
    \RightLabel{\rm ($\rig \supset$)}
    \UnaryInfC{$ \cseqnes{\rig{A \supset B}} $}
    \DisplayProof
 \\[0.5cm]
  \AxiomC{$ \cseqnes{ \mblock{\Delta, \lef A } } $}
  \RightLabel{\rm ($\lef\square$)}
  \UnaryInfC{$ \cseqnes{\lef{\Box A}, \mblock{\Delta } } $}
  \DisplayProof
    & 
  \AxiomC{$ \cseqpdown{\mblock{\rig A}} $}
  \RightLabel{\rm ($\rig\square$)}
  \UnaryInfC{$ \cseqnes{\rig A} $}
  \DisplayProof\\[0.5cm]
  \AxiomC{$ \cseqnes{\mblock{\lef{ A}}} $}
    \RightLabel{\rm ($\lef\Diamond$)}
    \UnaryInfC{$ \cseqnes{\lef{\Diamond A}} $}
    \DisplayProof
 & 
  \AxiomC{$ \cseqnes{ \mblock{\Delta, \rig A} } $}
    \RightLabel{\rm ($\rig\Diamond$)}
    \UnaryInfC{$ \cseqnes{\rig{\Diamond A}, \mblock{\Delta} }  $}
    \DisplayProof\\[0.5cm]
    \AxiomC{$\cseqnes{\lef A, \lef A}$}
    \RightLabel{$(\lef{\mathsf{ctr}})$}
    \UnaryInfC{$\cseqnes{\lef A}$}
    \DisplayProof
    & 
     \AxiomC{$\cseqnes{\rig A, \rig A}$}
    \RightLabel{$(\rig{\mathsf{ctr}})$}
    \UnaryInfC{$\cseqnes{\rig A}$}
    \DisplayProof
    \end{tabular}
\end{center}
  \caption{Rules of $\nikm$}\label{fig:nikm}
  \end{figure}

To define the flat bi-nested sequent corresponding to a polarised nested sequent in presence of contexts, we need to make sure that the correct notion of context is used. To this aim, we rewrite a polarised nested sequent in such a way that it can be immediately translated. 
We  need to introduce some additional definitions.  First, it is convenient to reason about trees associated to sequents:

\begin{definition}
    Given a polarised nested sequent $\Sigma$, we write $\tree{\Sigma}$ to denote its \emph{sequent tree}, whose nodes, denoted by $\pi, \sigma,\dots$, are multisets of polarised formulas: 
    \begin{center}
        \begin{tikzpicture}[scale=0.8]
         \node[] (0) at (-7,1) [] {$\tree{\Sigma} =$};
       
        \node[] (1) at (0,0) [] {$\lef A_1, \dots \lef A_k, \rig B_1, \dots, \rig B_h $};
        \node[] (2) at (-4,2) [] {$\tree{\Sigma_1}$};
        \node[] (3) at (-1,2) [] {$\tree{\Sigma_2}$};
        \node[] (4) at (1,2) [] {$\dots$};
        \node[] (5) at (4,2) [] {$\tree{\Sigma_n}$};
        
        \draw[] (1) -- (2);
        \draw[] (1) -- (3);
        \draw[] (1) -- (5);
        
        \end{tikzpicture}
    \end{center}
    The \emph{depth} of a nested sequent is defined as the depth of its tree. 
\end{definition}

Given a polarised nested sequent $\cseqnes{\Pi}$, we shall now define its \emph{contextualised variant} as the polarised nested sequent $\Sigma'\{\Pi'\}$. 
Intuitively, $\Sigma'\{\Pi'\}$ and $\cseqnes{\Pi}$ have the same sequent tree, and thus denote the same polarised nested sequent.  Sequent $\Sigma'\{\Pi'\}$ is just a notational variant of $\cseqnes{\Gamma}$, where we make sure that the sequent $\Pi'$ which fills the hole $\{\,\}$ in the context $\Sigma'\{ \, \}$ contains the full subtree of $\tree{\cseqnes{\Gamma}}$ starting at the node where   $\{\,\}$ occurs. 


\begin{definition}
    For any polarised context  $\cseqnes{\,}$ 
    let $\{\}, \Xi = \pi \in \tree{\cseqnes{}} $, for $\Xi$ (possibly empty) multiset of polarised formulas.
    For any polarised nested sequent $\cseqnes{\Pi}$, 
    its \emph{contextualised variant}, denoted by 
    $ \Sigma^{\Pi}_{a}\{\Pi', \Sigma^\Pi_c \} $, is the unique polarised nested sequent such that: 
    \begin{itemize}
        \item $\tree{\cseqnes{\Pi}} = \tree{\Sigma^{\Pi}_{a}\{\Pi', \Sigma^\Pi_c \}} $; \\[-0.3cm]
        \item 
        $ \Pi'= \fm{\Pi}, \Xi = \pi$; \\[-0.3cm]
        \item $\Sigma^\Pi_c$, the \emph{children sequent}, is the polarised nested sequent containing all the nodes $\delta \in \tree{\cseqnes{\Pi}}$ such that 
        $\pi \notin \tree{\Pi}$ and 
        there is a path from $\pi$ to $\delta$ in $\tree{\cseqnes{\Pi}}$;   \\[-0.3cm]  
        \item $\Sigma^{\Pi}_{a}\{\} $, the \emph{ancestor context}, is the polarised context 
        containing all the nodes nodes $\delta \in \tree{\cseqnes{\Pi}}$ such that 
        $\delta \notin \tree{\Pi}$ and $\delta \notin \tree{\Sigma^\Pi _c}$, and by setting 
         $\pi = \{\} $. 
    \end{itemize}
\end{definition}


\begin{example}
\label{example:extended:context}
    Take $\cseqnes{} = \lef A, \lef B, \mblock{\lef C, \rig D},\mblock{\{\}, \rig{E}, \mblock{\rig{F}}}$ and $\Pi = \lef H, \mblock{\lef J}$. Then:
    $$
    \cseqnes{\Pi} = \lef A, \lef B, \mblock{\lef C, \rig D},\mblock{\lef H, \mblock{\lef J}, \rig{E}, \mblock{\rig{F}}}
    $$
    The extended sequent of $\cseqnes{\Pi}$ is defined by taking $\Sigma^\Pi_c = \mblock{\rig F}$, $\Xi= \rig E$, and   $\Sigma^\Pi_a\{\} = \lef A, \lef B, \mblock{\lef C, \rig D},\mblock{\{ \} }$. We have  $\Sigma^\Pi_a\{ \fm{\Pi}, \Xi, \Sigma^\Gamma_c\} = \Sigma^\Gamma_a\{ \lef H, \mblock{\lef J}, \rig E, \mblock{\rig F}\} = \cseqnes{\Gamma}$. 
\end{example}

We are now ready to define the translation from polarised nested sequents into flat bi-nested sequents. 
Given a polarised nested sequent $ \Pi$, let $\lef \Pi = \{ A \mid \lef A \in \Pi \}$ and $\rig \Pi = \{ B \mid \rig B \in \Pi \}$. Moreover, we denote by by $\Pi^{[\,]}$ the multiset of  $[\,]$-blocks occurring in $\Pi$. 


\begin{definition}
\label{def:translation:multi}
    Let $\cseqnes{\Pi}$ be a polarised nested sequent. We define its corresponding flat bi-nested sequent, denoted by $\fl{\cseqnes{\Pi}}$, by induction on the depth $d$ of  $\cseqnes{\Pi}$: 
    \begin{itemize}
        \item If $d = 0$, then 
        $\cseqnes{\Pi} = \Sigma^\Pi_a\{ \Pi', \Sigma^\Pi_c\} =  \Pi'$, for $\Pi'$ multiset of polarised formulas. 
        We set: 
        $$
        \fl{\cseqnes{\Pi}} = \fl{\Sigma^\Pi_a\{ \Pi', \Sigma^\Pi_c\}} = \fl{\Pi' } := \lef{(\Pi')} \Rightarrow \rig {(\Pi')};
        $$
        \item If $d >0$, then 
        $\cseqnes{\Pi} = \Sigma^\Pi_a\{ \Pi', \Sigma^\Pi_c\}$. We set:
        $$
        \fl{\cseqnes{\Pi}} = \fl{\Sigma^\Pi_a\{ \Pi', \Sigma^\Pi_c\}} : = \cseqp{ \Psi \Rightarrow \Theta }
        $$
        where: \\[-0.3cm]
        \begin{itemize}
            \item $ \cseqp{} = \fl{\Sigma^\Pi_a\{\emptyset\}}$; \\[-0.3cm]
            \item $\Psi=  \lef{(\Pi')}$;\\[-0.3cm]
            \item $\Theta = \rig{(\Pi')}, 
            \mblock{\fl{\Theta_1}}, \dots, \mblock{\fl{\Theta_n}}$ where to each  $\Theta_i  $ there corresponds a block $\mblock{\Theta_i}$ in $\Pi \cup \Sigma^\Pi_c$. 
        \end{itemize}
    \end{itemize}
\end{definition}




\begin{example}
    Consider the polarised nested sequent $\cseqnes{\Pi}$ from Example~\ref{example:extended:context} and its extended context $\Sigma^\Pi_a\{ \Pi', \Sigma^\Pi_c\}$. Then: 
    $$
    \fl{\cseqnes{\Pi}} = \fl{\Sigma^\Pi_a\{ \Pi', \Sigma^\Pi_c\}} = G\{ H \Rightarrow  E, \mblock{\fl{\lef{J}}}, \mblock{\fl{\rig F}} \}.
    $$
    Moreover, $\fl{\lef{J}} = J \Rightarrow \emptyset $ and  $\fl{\rig{F}} = \emptyset \Rightarrow F$, and $\cseqp{} = \fl{\Sigma^\Pi_a\{ \}} = A, B \Rightarrow \mblock{C \Rightarrow D},\mblock{\{ \} }$. Putting everything together:
    $$
    \fl{\cseqnes{\Pi}} =  A, B \Rightarrow \mblock{C \Rightarrow D}, \mblock{ H \Rightarrow  E, \mblock{J \Rightarrow \,}, \mblock{\, \Rightarrow F}}.
    $$
\end{example}

We are now ready to prove the main result of this section. 
Slightly abusing the notation, for $\cseqp{}$ bi-nested context we denote by $\cseqpstar{}$ the bi-nested context obtained by removing all formulas occurring in all the consequents of all sequents occurring in $\cseqp{}$. Similarly, for 
$S$ bi-nested sequent and $\Delta$ multiset of  $\mblock{\, }$-blocks, we denote by $\cstar{S}$ and $\cstar{\Delta}$ the bi-nested sequent and multiset of blocks obtained by removing all formulas from the consequents of all sequents occurring in ${S}$ and $\Delta$ respectively.


\begin{theorem}
\label{theorem:nested-bi-nested}
    Let $\mathcal{D}$ be a proof in $\nikm$ of the polarised nested sequent $\Sigma$. Then, there is a proof $\mathcal{D}'$ in \cik of the flat bi-nested sequent $\fl{\Sigma}$. 
\end{theorem}
\begin{proof}
    We proceed by induction on the structure of $\mathcal{D}$, and show how to simulate in \cik every rule applied in $\mathcal{D}$. 
    In the base case, $\Sigma$ is an instance of $(\lef\bot)$ or $(\id)$. Then $\fl{\Sigma}$ is an instance of $(\lr{\bot})$ or $(\id)$, respectively. 
    For the inductive step, we distinguish cases according to the last rule $(r)$ applied in $\mathcal{D}$. Each such rule has the following form:
    \begin{center}
        \AxiomC{$\{\Sigma_i\}_{i \in \{1,2\}}$}
        \RightLabel{$(r)$}
        \UnaryInfC{$\Sigma$}
        \DisplayProof
    \end{center}
    We distinguish cases according to the last rule applied. 
    By inductive hypothesis, we have that for all $i$, $\fl{\Sigma_i}$ is derivable in \cik. We need to show how to derive the flat bi-nested sequent $\fl{\Sigma}$. All the cases, except for those of $(\rig{\supset})$ and $(\rig{\Box})$,  are quite easy and just employ the corresponding rule of \cik. We only show the case of rule $(\lef{\Box})$: 
    \begin{center}
    \AxiomC{$ \cseqnes{ \mblock{A, \Delta}} $}
    \RightLabel{\rm ($\lef \Box$)}
    \UnaryInfC{$  \cseqnes{\Box A, \mblock{ \Delta}} $}
    \DisplayProof
    \end{center}
    Suppose that $\fl{\cseqnes{ \mblock{A, \Delta}}} = \cseqp{\Psi_1\Rightarrow \Theta_1, \mblock{A, \Psi_2 \Rightarrow \Theta_2}}$, where $\Theta_1$ and $\Theta_2$ are multisets of formulas and $\mblock{\cdot}$-blocks. By inductive hypothesis, the sequent is derivable in \cik. We construct the following derivation:
    \begin{prooftree}
    \AxiomC{$ \cseqp{\Psi_1 \Rightarrow \Theta_1, \mblock{A, \Psi_2 \Rightarrow \Theta_2}} $}
    \RightLabel{$\wk$}
    \UnaryInfC{$ \cseqp{\Psi_1 , \Box A \Rightarrow \Theta_1, \mblock{A, \Psi_2 \Rightarrow \Theta_2}} $}
    \RightLabel{\rm ($\lr \Box$)}
    \UnaryInfC{$ \cseqp{\Psi_1, \Box A \Rightarrow \Theta_1, \mblock{ \Psi_2 \Rightarrow \Theta_2}}  $}
    \end{prooftree}
    We have that $\fl{\cseqnes{\Box A, \mblock{ \Delta}}} = \cseqp{\Psi_1, \Box A \Rightarrow \Theta_1, \mblock{ \Psi_2 \Rightarrow \Theta_2}}$. 

    \
    
    The cases of rules  $(\rig \supset)$ and $(\rig \Box)$ are much more involved because, starting from the premiss, we need to ``reconstruct'' the consequents of every component of the bi-nested sequent. We only show the case of  $(\rig \supset)$, the other being similar: 
    \begin{center}
    \AxiomC{$ \cseqpdown{\lef A, \rig B} $}
    \RightLabel{\rm ($\rig \supset$)}
    \UnaryInfC{$ \cseqnes{\rig{A \supset B}} $}
    \DisplayProof
    \end{center}
    Suppose that 
    $\fl{\cseqnes{\rig{A} ,\lef{ B}} } =  \cseqpstar{\Psi , A\Rightarrow  B, \Theta^*}$, for $\Psi$ multiset of formulas and $\Theta$ multiset of $\mblock{\,}$-blocks with no formulas occurring in their consequents. 
    Intuitively, $\Theta^*$ collects all the nodes that are children of the node where $A\supset B$ occurs, while $\cseqpstar{\,}$ collects all the other nodes in the flat sequent. 
    Observe that no formulas occur in the consequents of blocks in $\Theta^*$. 
    Then, we have  $\fl{\cseqnes{\rig{A \supset B}} } =  \cseqp{\Psi \Rightarrow A\supset  B, \Theta}$. This bi-nested sequent is obtained by ``adding'' the relevant formulas in the consequents of blocks in  $\cseqpstar{\,}$ and $\Theta^*$. 
    We shall construct the following \cik derivation:
    \begin{prooftree}
        \AxiomC{$\cseqpstar{\Psi , A\Rightarrow  B, \Theta^*}$}
        \noLine 
        \UnaryInfC{$\vdots$}
        \noLine 
        \UnaryInfC{$\cseqp{\Psi \Rightarrow A\supset  B, \Theta}$}
    \end{prooftree} 
    The derivation comprises several steps, depending on the shape of $\cseqpstar{}$ and $\Theta^*$. 
    The idea is to ``copy'' the tree structure encoded by  the sequent $\cseqpstar{\Psi , A\Rightarrow  B, \Theta^*}$ in the sequent $\cseqp{\Psi \Rightarrow A\supset  B, \Theta}$, while adding the relevant formulas in the consequents. 
    To this aim, 
    after the initialisation steps (Steps~\ref{it:case1}, \ref{it:case2} below) 
    we use \rulebc to ``transform'' $\cseqpstar{}$ into $\cseqp{}$ (Step~\ref{it:case3} below).
    In this stage, \rulefc might also need to be applied, to account for other branches of the sequent / tree starting at ancestors of the node where $A \supset B$ occurs.  Step~\ref{it:case2} can be skipped if $\cseqpstar{}$ is empty (this is the case distinction made in Step~\ref{it:case1}). 
    Then, in Step~\ref{it:case4}, we ``transform'' $\cstar{\Theta}$ into $\Theta$, using \rulefc. A final application of $(\rr \supset)$ yields the desired sequent. 
    
    Before detailing  how to construct the derivation, we need some additional definitions. 
    For some  $n \geq 0$, we have:  
        $$
        \cseqp{\Psi, A \Rightarrow B, \Theta} 
        \,
        = 
        \,
        \Psi_1 \Rightarrow \Theta_1, \mblock{\Psi_2 \Rightarrow \Theta_2, \mblock{ \cdots \mblock{\Psi_n \Rightarrow \Theta_n, \mblock{\Psi, A \Rightarrow B, \Theta}} \cdots }}
        $$
    Let:  
    \begin{eqnarray*}
        \cseqnum{0}{}&  = & \{ \}\\
        \cseqnum{1}{}& = &  \Psi_1 \Rightarrow \Theta_1, \mblock{\{ \}}\\
        \cseqnum{2}{} & = & \Psi_1 \Rightarrow \Theta_1, \mblock{ \Psi_2 \Rightarrow \Theta_2, \mblock{\{ \}}}\\
        & \vdots & \\
        \cseqnum{n}{}&  = &  \Psi_1 \Rightarrow \Theta_1, \mblock{\Psi_2 \Rightarrow \Theta_2, \mblock{ \cdots  \mblock{\Psi_n \Rightarrow \Theta_n, \mblock{ \{\} }} \cdots }}
    \end{eqnarray*}
        Moreover, set: 
    \begin{eqnarray*}
        \seqstar{n+1} & = &  \Psi, A \Rightarrow B, \cstar{\Theta}\\
        \seqstar{n} & = &  \Psi_{n} \Rightarrow \cstar{\Theta_{n}}, \mblock{\seqstar{n+1}}\\
        \seqstar{n-1} & = &  \Psi_{n-1} \Rightarrow \cstar{\Theta_{n-1}}, \mblock{\seqstar{n}}\\
        & \vdots &\\
        \seqstar{1}& = & \Psi_{1} \Rightarrow \cstar\Theta_1, \mblock{\seqstar{2}}
    \end{eqnarray*}    

    We now describe the steps to construct the derivation. We start from sequent $\cseqpstar{\Psi, A \Rightarrow B, \Theta^*}   
    = \seqstar{1}$.    
    
    \begin{enumerate}
        \item 
        \label{it:case1}
        If $n = 0$, then $\seqstar{1} =\Psi, A \Rightarrow B, \Theta^* = \seqstar{n+1}$. Construct the following derivation, where 
        the dashed line is a rewriting step:
        \begin{center}
            \AxiomC{$\Psi, A \Rightarrow B, \Theta^* $}
            \RightLabel{$(w_{C})$}
            \UnaryInfC{$\Psi \Rightarrow B, \Theta, \iblock{\Psi, A \Rightarrow B, \Theta^* }$}
            \dashedLine
            \UnaryInfC{$\cseqnum{0}{\Psi \Rightarrow A \supset B, \Theta,   \iblock{ \Psi, A \Rightarrow B, \Theta^* }}$}
            \DisplayProof
        \end{center}
        
        Then, continue the construction with Step~\ref{it:case4}. \\[-0.3cm]
               
        \item 
        \label{it:case2}
        Otherwise, if $n > 0$, construct the following derivation:
        \begin{center}
            \AxiomC{$\seqstar{1}$}
            \RightLabel{$(w_{C})$}
            \UnaryInfC{$\Psi_1\Rightarrow \Theta_1,  \iblock{\seqstar{1}}$}
            \dashedLine
            \UnaryInfC{$ \cseqnum{0}{ \Psi_1\Rightarrow \Theta_1,  \iblock{\seqstar{1}}} $}
            \DisplayProof
        \end{center}
        
        Take $i = 0$ and continue with Step~\ref{it:case3}.\\[-0.3cm] 

    \begin{figure}[t!]
        \begin{tabular}{c | c }
        \begin{tikzpicture}[scale=0.8]
        \tikzstyle{node}=[circle,fill=black,inner sep=1.2pt]
        
        \node[label= left :{\small{$\Psi_{i+1} \Rightarrow \Theta'_{i+1} $}}] (l1) at (0,-0.7) [node] {};
        \node[label= left :{\small{$\Psi_{i+1} \Rightarrow \, $}}] (u1) at (0, 2) [node]
        {};
        \node[label=  right :{\small{$\seqstar{i+2}  \textcolor{white}{\Rightarrow \Theta_{i+2}}$}}] (u2) at (2.7,2) [node]
        {};
        \node[label=  above right  :{\small{$\Lambda_1^*$}}] (uc1) at (1,3) [node] {};
        \node[label=  right  :{\small{$\Lambda_k^*$}}] (ucn) at (1.5,2.6) [node] {};
        \node[label=  above right  :{\small{$\Lambda_1$}}] (lc1) at (1,0.3) [node] {};
        \node[label=  right  :{\small{$\Lambda_k$}}] (lcn) at (1.5,-0.1) [node] {};
        
        \draw[->, dashed] (l1) -- (u1);
        \draw[->, ] (u1) -- (u2);
        \draw[->, ] (u1) -- (uc1);
        \draw[->, ] (u1) -- (ucn);
        \draw[dotted] (uc1) -- (ucn);
        \draw[->, ] (l1) -- (lc1);
        \draw[->, ] (l1) -- (lcn);
        \draw[dotted] (uc1) -- (ucn);
        \draw[dotted] (lc1) -- (lcn);

        \end{tikzpicture}
             &  
        \begin{tikzpicture}[scale=0.8]
        \tikzstyle{node}=[circle,fill=black,inner sep=1.2pt]
        
        \node[label= left :{\small{$\Psi_{i+1} \Rightarrow \Theta'_{i+1} $}}] (l1) at (0,-0.7) [node] {};
        \node[label= left :{\small{$\Psi_{i+1} \Rightarrow \, $}}] (u1) at (0, 2) [node]
        {};
        \node[label=  right :{\small{$\seqstar{i+2}  \textcolor{white}{\Rightarrow \Theta_{i+2}}$}}] (u2) at (2.7,2) [node]
        {};
        \node[label=  above right  :{\small{$\Lambda_1$}}] (uc1) at (1,0.3) [node] {};
        \node[label=  right  :{\small{$\Lambda_k$}}] (ucn) at (1.5,-0.1) [node] {};
        
        \draw[->, dashed] (l1) -- (u1);
        \draw[->, ] (u1) -- (u2);
        \draw[->, ] (l1) -- (uc1);
        \draw[->, ] (l1) -- (ucn);
        \draw[dotted] (uc1) -- (ucn);
        \end{tikzpicture}
        \\[0.5cm]
        \hline 
        & \\[0.5cm]
        \begin{tikzpicture}[scale=0.8]
        \tikzstyle{node}=[circle,fill=black,inner sep=1.2pt]
        
        \node[label= left :{\small{$\Psi_{i+1} \Rightarrow \Theta'_{i+1} $}}] (l1) at (0,-0.7) [node] {};
        \node[label= left :{\small{$\, \Rightarrow \, $}}] (u1) at (0, 2) [node] {};
        \node[label=   right :{\small{$\seqstar{i+2} $}}] (u2) at (2.7,2) [node]
        {};
        \node[label=   right :{\small{$\seqstar{i+2}$}}] (uu2) at (2.7,1.5) [node]
        {};
        \node[label=  above right  :{\small{$\Lambda_1$}}] (uc1) at (1,0.3) [node] {};
        \node[label=  right  :{\small{$\Lambda_k$}}] (ucn) at (1.5,-0.1) [node] {};
        \node[label= right :{\small{$\Psi_{i+2} \Rightarrow \Theta_{i+2} $}}] (l2) at (2.7, -0.7) [node] {};

        \draw[->, dashed] (l1) -- (u1);
        \draw[->, dashed] (l2) -- (uu2);
        \draw[->, ] (l1) -- (l2);
        \draw[->, ] (u1) -- (u2);
        \draw[->, ] (l1) -- (uc1);
        \draw[->, ] (l1) -- (ucn);
        \draw[dotted] (uc1) -- (ucn);
        \end{tikzpicture}
             & 
           
        \begin{tikzpicture}[scale=0.8]
        \tikzstyle{node}=[circle,fill=black,inner sep=1.2pt]
        
        \node[label= left :{\small{$\Psi_{i+1} \Rightarrow \Theta'_{i+1} $}}] (l1) at (0,-0.7) [node] {};
        \node[label=   right :{\small{$\seqstar{i+2}$}}] (uu2) at (2.7,2) [node]
        {};
        \node[label=  above right  :{\small{$\Lambda_1$}}] (uc1) at (1,0.3) [node] {};
        \node[label=  right  :{\small{$\Lambda_k$}}] (ucn) at (1.5,-0.1) [node] {};
        \node[label= right :{\small{$\Psi_{i+2} \Rightarrow \Theta_{i+2}$ }}] (l2) at (2.7, -0.7) [node] {};

        \draw[->, dashed] (l2) -- (uu2);
        \draw[->, ] (l1) -- (l2);
        \draw[->, ] (l1) -- (uc1);
        \draw[->, ] (l1) -- (ucn);
        \draw[dotted] (uc1) -- (ucn);
        \end{tikzpicture}\\[0.5cm]
        \end{tabular}
        \caption{A graphical representation of one iteration of Step~\ref{it:case3}. For simplicity, the contexts are not represented. 
        Dashed lines represent $\iblock{\cdot}$-blocks, while solid lines represent $\mblock{\cdot}$-blocks. The four pictures depict the evolution of the sequent within the derivation in Step~\ref{it:case3}. 
        The  \textbf{top-left} sequent is the starting point of the derivation. The sequent on \textbf{top-right} is obtained from this sequent by applying \rulefc $k$ times. Then, the  \textbf{bottom-left} sequent is obtained  applying  {\rm \text{(trans)}} and weakening steps to the top-right sequent. With one final application of \rulebc we obtain the \textbf{bottom-right} sequent, which is the conclusion of the derivation. The sequent $\Psi_{i+1} \Rightarrow \Theta'_{i+1}, \mblock{\Lambda_1}, \dots, \mblock{\Lambda_k} $ now becomes part of the context, and we can repeat the iteration with sequent $\Psi_{i+2} \Rightarrow \Theta_{i+2}, \iblock{\seqstar{i+2}}$. 
        }
        \label{fig:step3}
    \end{figure}
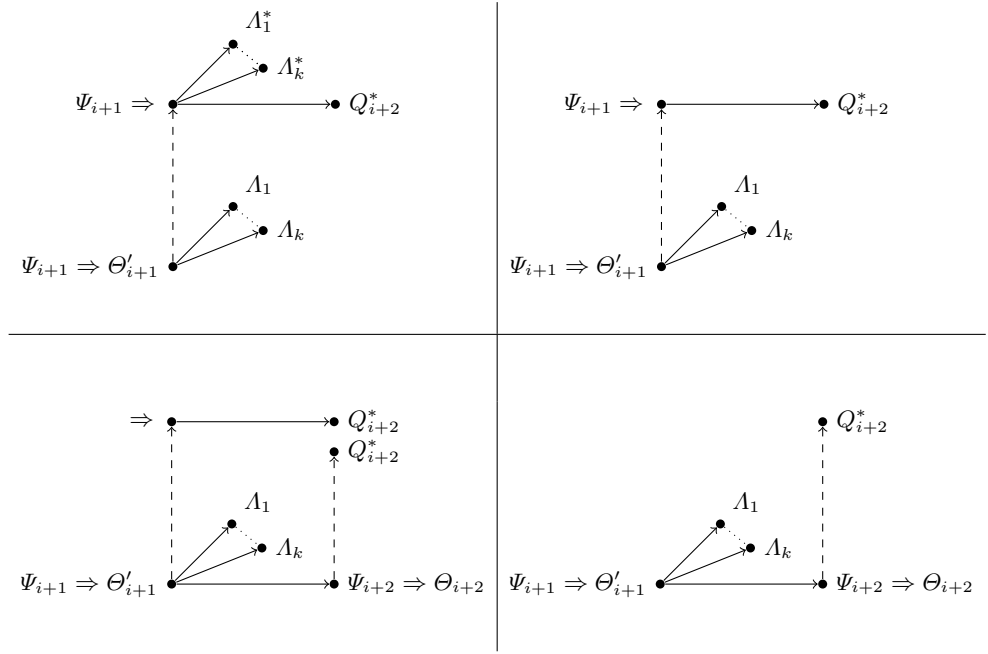

        \item 
        \label{it:case3}
        For any $0 \leq i < n$, we have that,  for $k\geq 0$, $\Theta_{i+1} = \Theta'_{i+1}, \mblock{\Lambda_1},\dots,  \mblock{\Lambda_k}$, where $\Theta'_{i+1}$ is a multiset of formulas. We construct the following derivation, also depicted in Figure~\ref{fig:step3}. In the derivation,  \rulefc$\times k$ denotes $k$ applications of \rulefc, while $\wk$ shortens possibly multiple applications of $(w_{R})$ and $(w_{L})$:
        \begin{center}
            \AxiomC{$G_i\{ \Psi_{i+1} \Rightarrow  {\Theta_{i+1}},\iblock{\seqstar{i+1} }\}$}
            \dashedLine
            \UnaryInfC{$G_i\{ \Psi_{i+1} \Rightarrow  {\Theta_{i+1}},\iblock{\Psi_{i+1} \Rightarrow \cstar{\Theta_{i+1}},\mblock{\seqstar{i+2}}}\}$}
            \dashedLine
            \UnaryInfC{$G_i\{\Psi_{i+1} \Rightarrow \Theta'_{i+1}, \mblock{\Lambda_1},\dots,  \mblock{\Lambda_k}, \iblock{\Psi_{i+1} \Rightarrow \mblock{\Lambda_1^*},\dots,  \mblock{\Lambda_k^*},{\mblock{\seqstar{i+2}}} }\}$}
            \RightLabel{\rulefc $ \times k$}
            \UnaryInfC{$G_i\{\Psi_{i+1} \Rightarrow \Theta'_{i+1}, \mblock{\Lambda_1},\dots,  \mblock{\Lambda_k}, \iblock{\Psi_{i+1} \Rightarrow \mblock{\seqstar{i+2}}}\}$}
            \dashedLine
            \UnaryInfC{$G_i\{\Psi_{i+1} \Rightarrow \Theta_{i+1},  \iblock{\Psi_{i+1} \Rightarrow \mblock{\seqstar{i+2}}}\}$}
            \RightLabel{\rm \text{(trans)}}
            \UnaryInfC{$G_i\{\Psi_{i+1} \Rightarrow \Theta_{i+1},  \iblock{\, \Rightarrow \mblock{\seqstar{i+2}}}\}$}
            \RightLabel{$\mathsf{wk}$}
            \UnaryInfC{$G_i\{\Psi_{i+1} \Rightarrow \Theta_{i+1},  \iblock{\, \Rightarrow \mblock{\seqstar{l}}}, \mblock{\Psi_{i+2} \Rightarrow \Theta_{i+2}, \iblock{\seqstar{i+2}}} \}$}
            \RightLabel{\rulebc}
            \UnaryInfC{$G_i\{\Psi_{i+1} \Rightarrow \Theta_{i+1},  \mblock{\Pi_{i+2} \Rightarrow \Theta_{i+2}, \iblock{\seqstar{i+2}}} \}$}
            \dashedLine
            \UnaryInfC{$G_{i+1}\{ \Psi_{i+2} \Rightarrow \Theta_{i+2}, \iblock{\seqstar{i+2}} \}$}
            \DisplayProof
        \end{center}
        Repeat this step. In the last iteration of this step, when $i = n-1$, the input sequent will be $G_{n-1}\{ \Psi_{n} \Rightarrow  {\Theta_{n}},\iblock{\seqstar{n} }\} = 
        G_{n-1}\{ \Psi_{n} \Rightarrow  {\Theta_{n}},\iblock{\Psi_{n} \Rightarrow \cstar{\Theta_{n}},\mblock{\seqstar{n+1}}}\} 
        = G_{n-1}\{ \Psi_{n} \Rightarrow  {\Theta_{n}},\iblock{\Psi_{n} \Rightarrow \cstar{\Theta_{n}},\mblock{\Psi, A \Rightarrow B, \cstar{\Theta}}}\}$. In the $\wk$ application, the multiset introduced top-down are $\Pi_{n+1} = \Gamma$ and $\Theta_{n+1} = A \supset B, \Delta$. The lowermost sequent of the derivation is $ G_{n-1}\{\Pi_{n} \Rightarrow \Theta_{n},  \mblock{\Gamma \Rightarrow A \supset B, \Delta, \iblock{\Gamma, A \Rightarrow B, \cstar{\Delta} }}\}$ or, equivalently, sequent $ G_{n}\{\Psi \Rightarrow A \supset B, \Theta, \iblock{\Psi, A \Rightarrow B, \cstar{\Theta} }\}$. Continue the construction with Step~\ref{it:case4}. \\[-0.3cm]
        
        \item 
        \label{it:case4}
        Take sequent $ G_{h}\{\Psi \Rightarrow A \supset B, \Theta, \iblock{\Psi, A \Rightarrow B, \cstar{\Theta} }\}$, where $h \in \{0, n\}$. 
        For some $k \geq 0$, we have that $\Theta = \Theta', \mblock{\Lambda_1}, \dots, \mblock{\Lambda_k}$, where $\Delta'$ is a multiset of formulas. We construct the following derivation: 
        \begin{center}
        \AxiomC{$G_{h}\{\Psi \Rightarrow A \supset B, \Theta, \iblock{\Psi, A \Rightarrow B, \cstar{\Theta} }\}$}
       \dashedLine
        \UnaryInfC{$G_{h}\{\Psi \Rightarrow A \supset B, \Theta', \mblock{\Lambda_1}, \dots, \mblock{\Lambda_k}, \iblock{\Psi, A \Rightarrow B,  \mblock{\cstar\Lambda_1}, \dots, \mblock{\cstar\Lambda_k} }\}$} 
        \RightLabel{\rulefc $\times k$}
        \UnaryInfC{$G_{h}\{\Psi \Rightarrow A \supset B, \Theta', \mblock{\Lambda_1}, \dots, \mblock{\Lambda_k}, \iblock{\Psi, A \Rightarrow B }\}$}
        \RightLabel{$\mathsf{trans}$}
    \UnaryInfC{$G_{h}\{\Psi \Rightarrow A \supset B, \Theta', \mblock{\Lambda_1}, \dots, \mblock{\Lambda_k}, \iblock{ A \Rightarrow B }\}$}
    \RightLabel{$(\rr{\supset})$}
    \UnaryInfC{$G_{h}\{\Psi \Rightarrow A \supset B, \Theta', \mblock{\Lambda_1}, \dots, \mblock{\Lambda_k}\}$}
    \dashedLine
    \UnaryInfC{$G_{h}\{\Psi \Rightarrow A \supset B, \Theta\}$}
    \DisplayProof
\end{center}
    For both $h = 0$ and $h = n$ we have that $G_{h}\{\Psi \Rightarrow A \supset B, \Theta\} =  \cseqp{\Psi, A \Rightarrow B, \Theta}$, and we have concluded the construction.  
    \end{enumerate}
\end{proof}

\begin{example}
\label{ex:nested}
Consider the following $\nikm$ derivation:
    \begin{prooftree}
        \AxiomC{}
        \RightLabel{$(\lef{\bot})$}
        \UnaryInfC{$\lef{\bot}, \dots$}
        \AxiomC{}
        \RightLabel{\rm ($\text{id}$)}
        \UnaryInfC{$\dots, \mblock{\lef p,\lef{\bot}}$}
        \AxiomC{}
        \RightLabel{}
        \UnaryInfC{$  \rig{\bot}, \mblock{\lef{p}, \rig{p}}$}
        \RightLabel{$(\lef{\supset})$}
        \BinaryInfC{$  \rig{\bot}, \mblock{\lef{p}, \lef{p\supset \bot}}$}
        \RightLabel{$(\lef{\Box})$}
        \UnaryInfC{$ \lef{\Box (p \supset \bot)}, \rig{\bot}, \mblock{\lef{p}}$}
        \RightLabel{$(\rig{\supset})$}
        \UnaryInfC{$ \rig{\Box (p \supset \bot)\supset \bot},  \mblock{\lef{p},\rig \bot}$}
        \RightLabel{$(\lef{\supset})$}
        \BinaryInfC{$ \lef{(\Box (p \supset \bot)\supset \bot)\supset \bot},  \mblock{\lef{p},\rig \bot}$}
        \RightLabel{$(\rig{\supset})$}
        \UnaryInfC{$ \lef{(\Box (p \supset \bot)\supset \bot)\supset \bot},  \mblock{\rig{p \supset \bot}}$}
        \RightLabel{$(\rig{\Box})$}
        \UnaryInfC{$ \lef{(\Box (p \supset \bot)\supset \bot)\supset \bot}, \rig{\Box(p \supset \bot)}$}
        \RightLabel{$(\rig{\supset})$}
        \UnaryInfC{$ \rig{((\Box (p \supset \bot)\supset \bot)\supset \bot) \supset \Box(p \supset \bot)}$}
    \end{prooftree}
    \newcommand{\mseq}{$\textcolor{red}{*}$\quad}
    The translation of the rightmost branch of the above derivation is showcased in Figure~\ref{fig:example:tr}.  
    \begin{figure}[t!]
    \begin{prooftree}
        \AxiomC{}
        \RightLabel{(id)}
        \UnaryInfC{\mseq $\Rightarrow \bot, \mblock{p \Rightarrow p}$}
        \RightLabel{$\wk$}
        \UnaryInfC{$\Rightarrow \bot, \mblock{ p, p \supset \bot \Rightarrow p}$}
        \RightLabel{$(\lr{\supset})$}
        \UnaryInfC{\mseq$\Rightarrow \bot, \mblock{p, p \supset \bot \Rightarrow }$}
        \RightLabel{$\wk$}
        \UnaryInfC{$\Box (p \supset \bot)\Rightarrow \bot, \mblock{p, p \supset \bot \Rightarrow }$}
        \RightLabel{$(\lr{\Box})$}
        \UnaryInfC{\mseq $\Box (p \supset \bot)\Rightarrow \bot, \mblock{ p \Rightarrow }$}
        \RightLabel{$\wk$}
        \UnaryInfC{ $ \Rightarrow 
        \mblock{p \Rightarrow \bot},  \iblock{\Box (p \supset \bot)\Rightarrow \bot, \mblock{ p \Rightarrow }} $}
        \RightLabel{\rulefc}
        \UnaryInfC{ $ \Rightarrow 
        \mblock{p \Rightarrow \bot},  \iblock{\Box (p \supset \bot)\Rightarrow \bot} $}
        \RightLabel{$(\rr{\supset})$}
        \UnaryInfC{\mseq $ \Rightarrow \Box(p \supset \bot) \supset \bot, \mblock{p \Rightarrow \bot} $}
        \RightLabel{$\wk$}
        \UnaryInfC{ $ F \Rightarrow \Box(p \supset \bot) \supset \bot, \mblock{p \Rightarrow \bot} $}
        \RightLabel{$(\lr{\supset})$}
        \UnaryInfC{\mseq $ F \Rightarrow \mblock{p \Rightarrow \bot} $}
        \RightLabel{$\wk$}
        \UnaryInfC{ $\Rightarrow \iblock{F \Rightarrow \mblock{p \Rightarrow \bot}} $}
        \RightLabel{(trans)}
        \UnaryInfC{ $F \Rightarrow \mblock{\Rightarrow  \iblock{   p \Rightarrow \bot}}, \iblock{\Rightarrow \mblock{p \Rightarrow \bot}} $}
        \RightLabel{\rulebc}
        \UnaryInfC{ $F \Rightarrow \mblock{\Rightarrow  \iblock{   p \Rightarrow \bot}} $}
        \RightLabel{$(\rr{\supset})$}
        \UnaryInfC{\mseq $F \Rightarrow \mblock{\Rightarrow    p \supset \bot} $}
        \RightLabel{$\wk$}
        \UnaryInfC{ $F \Rightarrow \iblock{F \Rightarrow \mblock{\Rightarrow    p \supset \bot}} $}
        \RightLabel{(trans)}
        \UnaryInfC{ $F \Rightarrow \iblock{ \Rightarrow \mblock{\Rightarrow    p \supset \bot}} $}
        \RightLabel{$(\rr{\Box})$}
        \UnaryInfC{\mseq $F \Rightarrow \Box (p \supset \bot) $}
        \RightLabel{$\wk$}
        \UnaryInfC{ $  \Rightarrow \iblock{F \Rightarrow \Box (p \supset \bot)} $}
        \RightLabel{$(\rr{\supset})$}
        \UnaryInfC{ \mseq $ \Rightarrow F \supset \Box (p \supset \bot) $}
    \end{prooftree}
        \caption{Translation of the derivation from Example~\ref{ex:nested}, where $F$ abbreviates formula $(\Box (p \supset \bot)\supset \bot)\supset \bot$ and the sequents marked with \mseq\!\!\! directly translate nested sequents in the $\nikm$ derivation. Steps marked with $\wk$ correspond to possibly multiple applications of the weakening rules.}
        \label{fig:example:tr}
    \end{figure}
    
\end{example}

\

To conclude, we observe that every nested sequent is a flat bi-nested sequent, and thus it is trivial to translate a nested sequent to its bi-nested version (the additional steps we introduced served the purpose of taking care of contexts). However, bi-nested sequents with $\iblock{\cdot}$-structures cannot be immediately translated into nested sequents. We have not investigated this side of the translation, but we conjecture that there are more bi-nested sequents than nested sequents. Thus, bi-nested sequents stand somehow in between labelled and nested sequents, still encoding a tree-like structure but also allowing invertibility of all the rules thanks to the use of $\iblock{\cdot}$-blocks. 

The size of the \cik proofs obtained by translating $\nikm$ proofs is much bigger than the size of the original derivation. The nested  rules for $(\rig{\supset})$ and $(\rig{\Box})$ just ``copy'' the structure of the nested sequent in the conclusion in the premiss. To simulate this step in the bi-nested calculus, we need to go through all the nodes of the tree, and use \rulefc and \rulebc to ``copy'' the structure. This could have been avoided by defining rules for $\rr{\supset}$ and $\rr{\Box}$ that incorporate applications of  \rulefc and \rulebc. However, the proof of admissibility of such rules in \cik would need to go through a construction similar to the one in Theorem~\ref{theorem:nested-bi-nested}.

\subsection{Relations with other calculi}
\label{sec:simpson_gore}

There are other calculi for $\ik$ proposed in the literature, both label-free and labelled. 
In \cite{GorePT10} and \cite{postniece2010proof},  
label-free calculi 
are proposed for several extensions of intuitionistic modal logic comprising dual implication $\dimpl$ and tense modalities: the ``looking backward modalities" $\blacksquare$ and $\blacklozenge$ in addition to standard $\square$ and $\Diamond$. In these logics the pairs ($\blacksquare,\Diamond$) and ($\square,\blacklozenge$) form a Galois connection, while there is no relation between modalities of the same color. The  semantics of all modalities is defined by means of two independent relations $R_\Box$ and $R_\Diamond$ and their inverses. 

The calculi from \cite{GorePT10,postniece2010proof} make use of three structural connectives (
besides ``,"):  the monadic $\circ$ and $\bullet$, and the dyadic $\triangleright$. 
The interpretation of the connectives depends on their polarity, 
so that $\circ$  is interpreted  as $\Diamond$ or $\Box$, $\bullet$ as $\blacklozenge$ or $\blacksquare$ and $\triangleright$ as $\dimpl$ or $\supset$ according to the positive or negative polarity. 
Two kinds of calculi are proposed: a shallow (or display) calculus $\lbikt$ and a deep calculus $\dbikt$. In the former inference rules are applied only to top level structures and display rules are used to bring nested components to the surface; while in the latter, inference rules can be applied to arbitrary nested structures. Within this general framework the authors obtain a calculus for \ik in the \emph{shallow} form by 
removing the rules for $\dimpl$, $\blacklozenge$ and  $\blacksquare$ and 
introducing the two following rules to capture the specific semantic conditions of \ik:
\begin{center}
	\AxiomC{$X\Rightarrow \bullet (Y\triangleright\emptyset)\triangleright \bullet (\emptyset\triangleright Z)$}
	\RightLabel{($\bullet\triangleright_R$)}
	\UnaryInfC{$X\Rightarrow \bullet (Y\triangleright Z)$}
	\DisplayProof
	\quad 
	\AxiomC{$X\Rightarrow \circ (Y\triangleright\emptyset)\triangleright \circ (\emptyset\triangleright Z)$}
	\RightLabel{($\circ\triangleright_R$)}
	\UnaryInfC{$X\Rightarrow \circ (Y\triangleright Z)$}
	\DisplayProof
\end{center}
In \cite{postniece2010proof} it is argued also that the shallow calculus 
supports terminating proof search by adapting the general argument for the basic calculus $\lbikt$. 
 
There are some similarities
between the shallow calculus from \cite{GorePT10,postniece2010proof} and our calculus \cik. In particular, an implication block $\langle \Pi \Rightarrow \Sigma \rangle$ corresponds to $ \Pi \triangleright \Sigma $ with positive polarity, while a modal block $[\Pi \Rightarrow \Sigma]$ corresponds to $\circ (\Pi \triangleright \Sigma)$ with positive polarity. 
However, establishing a 
formal translation 
between the shallow calculus and our \cik, similarly to what is done in the previous subsections, 
appears to be non-trivial and is left for future investigations.

We end this section by recalling Simpson's calculus for \ik.  Simpson was indeed the first one to develop proof systems for \ik and its several extensions, namely for any modal logic whose accessibility relation is defined by \emph{geometric axioms}. In \cite{simpson1994proof}, he has proposed labelled calculi in the form of natural deduction and sequent calculus. Both calculi make use of world-labels and are obtained by the translation of \ik in first-order intuitionistic logic (FoIL), so that in particular:
\begin{itemize}
	\item $x\Vdash \Box A$ is interpreted as the FoIL formula $\forall y (Rx y \supset y\Vdash A) $
	\item $x\Vdash \Diamond A$ is interpreted as the FoIL formula $\exists y (Rx y \land y\Vdash A)$
\end{itemize}
where $\forall,\exists$ are treated as in FoIL. 
Consequently, as a difference with the fully labelled sequent calculus \labik by \cite{Marin:et:al:2021}, the calculus by Simpson employs single-conclusion sequents, and it is obtained by adopting the standard sequent rules for propositional IL, decorated with labels, and adding the following rules:  
\begin{center}
\begin{tabular}{cc}
   \AxiomC{$\Gamma, xRy \Rightarrow y: A$}
	\RightLabel{($\Box_R$)}
	\UnaryInfC{$\Gamma \Rightarrow x:\Box A$}
	\DisplayProof  
    &  
    \AxiomC{$\Gamma, x:\Box A,  xRy, x: A \Rightarrow z: C$}
	\RightLabel{($\Box_L$)}
	\UnaryInfC{$\Gamma, x:\Box A,  xRy\Rightarrow z: C$}
	\DisplayProof
    \\[.5cm]
    \AxiomC{$\Gamma, xRy\Rightarrow y: A$}
	\RightLabel{($\Diamond_R$)}
	\UnaryInfC{$\Gamma, x R y \Rightarrow x: \Diamond A$}
	\DisplayProof
    &
    \AxiomC{$\Gamma, xRy, y: A \Rightarrow z: C$}
	\RightLabel{($\Diamond_L$)}
	\UnaryInfC{$\Gamma, x:\Diamond A \Rightarrow z: C$}
	\DisplayProof
    \\
\end{tabular}
\end{center}
where $y$ is fresh in ($\Box_R$) and ($\Diamond_L$).

Simpson proves soundness and completeness of this calculus and shows how to obtain a terminating proof search, thereby establishing the decidability of \ik (and the decidability of several extensions of it, excluding the counterpart of K4 and S4). Although the issue is not dealt in \cite{Marin:et:al:2021}, it should be easy or at least possible to simulate Simpson's calculus by \labik. 
\section{Conclusion and future work}
\label{sec:conclusions}
We have proposed a bi-nested sequent calculus \cik for \ik. 
We have introduced a suitable proof search strategy resulting in a decision procedure for \ik and allowing for a direct countermodel construction. 
Furthermore, we have shown how our calculus can simulate proofs within the two main calculi for \ik. 

We believe that bi-nested calculi are 
a suitable 
formalism to cope with the family of intuitionistic and constructive modal logics. Similar calculi were given for related logics such as \fik (cf. \cite{csl-2024-fik}) and $\mathbf{LIK}$ (cf. \cite{balbiani2024local}). One advantage of bi-nested calculi  is that they provide a modular framework to capture all  these logic in a uniform way: logical rules are the same in all logics while interaction rules, which are structural rules, formalize specific semantic properties. 
While both labelled calculi and display calculi may provide a modular and general framework as well, but bi-nested calculi 
achieve this goal 
with minimal means, 
providing a significantly simpler and terminating proof-search procedure. 

In future research we want to investigate the complexity  of the decision procedure based on our calculus, as well as countermodel construction, knowing that both are estimated to be quite high for \ik. 
Moreover, we intend to extend the calculus to extensions of \ik by adding standard axioms from the modal cube, targeting logic $\mathbf{IS4}$, whose decidability has been recently established in \cite{Girlando:etal:2023}. 



\bibliographystyle{splncs04}
\bibliography{refs}

\end{document}